\documentclass[11pt,letter]{article}
\usepackage{graphicx} 
\usepackage{xcolor}

\title{Mechanism Design via the Interim Relaxation}

\author{Kshipra Bhawalkar\\
Google Research\\
\texttt{kshipra@google.com}
\and
Marios Mertzanidis \\
Purdue University\\
\texttt{mmertzan@purdue.edu}
\and
Divyarthi Mohan\\
Tel Aviv University\\
\texttt{divyarthim@tau.ac.il}
\and
Alexandros Psomas\\
Purdue University\\
\texttt{apsomas@cs.purdue.edu}}

\date{}

\input{Definitions}
\allowdisplaybreaks

\begin{document}

\maketitle

\begin{abstract}

We study revenue maximization for agents with additive preferences, subject to downward-closed constraints on the set of feasible allocations. 
In seminal work, Alaei~\cite{alaei2014bayesian} introduced a powerful multi-to-single agent reduction based on an ex-ante relaxation of the multi-agent problem. This reduction employs a rounding procedure which is an online contention resolution scheme (OCRS) in disguise, a now widely-used method for rounding fractional solutions in online Bayesian and stochastic optimization problems. In this paper, we leverage our vantage point, 10 years after the work of Alaei, with a rich OCRS toolkit and modern approaches to analyzing multi-agent mechanisms; we introduce a general framework for designing non-sequential and sequential multi-agent, revenue-maximizing mechanisms, capturing a wide variety of problems Alaei's framework could not address. Our framework uses an \emph{interim} relaxation, that is rounded to a feasible mechanism using what we call a two-level OCRS, which allows for some structured dependence between the activation of its input elements. For a wide family of constraints, we can construct such schemes using existing OCRSs as a black box; for other constraints, such as knapsack, we construct such schemes from scratch. We demonstrate numerous applications of our framework, including a sequential mechanism that guarantees a $\frac{2e}{e-1} \approx 3.16$ approximation to the optimal revenue for the case of additive agents subject to matroid feasibility constraints. We also show how our framework can be easily extended to multi-parameter procurement auctions, where we provide an OCRS for Stochastic Knapsack that might be of independent interest.
\end{abstract}

\newpage

\section{Introduction}


We consider the problem of a revenue-maximizing seller with $m$ heterogeneous items for sale to $n$ strategic agents with additive preferences, subject to downward-closed constraints on the set of feasible allocations. Revenue maximization for multi-agent environments is a central problem in Computer Science and Economics. Beyond Myerson's~\cite{myerson1981optimal} single-item mechanism, characterizing the revenue optimal mechanism in multi-item settings is a notoriously hard problem. Revenue-optimal mechanisms are hard to compute even in basic settings, and even exhibit various counter-intuitive
properties~\cite{manelli2007multidimensional,daskalakis2013mechanism,daskalakis2015strong,briest2015pricing,hart2015maximal,daskalakis2015multi,hart2019selling,psomas2022infinite}. An active research area strives to understand optimal and approximately optimal mechanisms from various perspectives, including their computational complexity~\cite{CDW_1, CDW_2, CDW_3, CDW_4},
sample complexity~\cite{cole2014sample,huang2015making,devanur2016sample,morgenstern2016learning,cai2017learning,guo2019settling,gonczarowski2021sample}, robustness~\cite{bergemann2011robust,cai2017learning,Dutting19,li2019revenue,psomas2019smoothed,brustle2020multi,makur2024robustness}, and the tradeoffs between simplicity and optimality~\cite{ChawlaHK07,ChawlaHMS10,ChawlaMS15,Yao15,RubinsteinW15,chawla2016subadditive,cai2017subadditive,kleinberg2019matroid,cai2019duality,babaioff2020simple,BabaioffGN17, KothariMSSW19}.


Influential work by Alaei~\cite{alaei2014bayesian} provides a framework for constructing multi-agent mechanisms with ex-post supply constraints via a reduction to single-agent mechanism design with ex-ante supply constraints.  On a high level, Alaei's framework first finds a feasible in expectation ex-ante allocation rule: a vector $x \in [0,1]^{nm}$, where $x_{i,j}$ is the probability of allocating item $j$ to agent $i$, over the randomness in the mechanism and all agents' valuations. Given this ex-ante relaxation, the framework needs a single agent mechanism for each agent $i$, such that item $j$ is allocated to agent $i$ with probability at most $x_{i,j}$. Alaei shows that running such single-agent mechanisms independently can be combined with a rounding procedure (in order to satisfy the supply constraints ex-post) to give an overall approximately optimal multi-agent mechanism. This rounding step, Alaei's solution to his ``magician's problem,'' is an online contention resolution scheme (henceforth, OCRS), in disguise. OCRSs, later defined by Feldman et al.~\cite{feldman2021online}, are a widely applicable tool for rounding fractional solutions in Bayesian and stochastic online optimization problems.

In this paper, we leverage our vantage point, 10 years after the work of Alaei~\cite{alaei2014bayesian}, with a rich OCRS toolkit and modern approaches to analyzing multi-agent mechanisms, to introduce a novel, general framework for designing both non-sequential and sequential multi-agent, revenue-maximizing mechanisms for agents with additive preferences, subject to downward-closed constraints on the set of feasible allocations. Our framework uses an \emph{interim} relaxation, that is rounded to a feasible mechanism, using what we call a two-level OCRS, allowing for some structured dependence between the activation of its input elements. For a wide family of constraints, we can construct such schemes using OCRSs as a black box; for other constraints, e.g., knapsack, we construct such schemes from scratch. We demonstrate numerous applications of our framework, including a sequential mechanism that guarantees a $\frac{2e}{e-1} \approx 3.16$ approximation to the optimal revenue for the case of additive agents subject to matroid feasibility constraints. We also show how our framework can be easily extended to multi-parameter procurement auctions, where we provide an OCRS for Stochastic Knapsack that might be of independent interest.

\subsection{Our Contributions}

Our framework relies on an interim relaxation.  Intuitively, an interim form (or reduced form) of a mechanism $\M$ has variables $\pi_{i,j}^{\M}(v_i)$, which indicate the probability that agent $i$ receives item $j$ when reporting valuation $v_i$ to $\M$ (over the randomness in $\M$, as well as the randomness in other agents' valuations), and variables $q_i^{\M}(v_i)$, which indicate the expected payment of agent $i$ when reporting valuation $v_i$ to $\M$ (over the same randomness). Writing a linear program that optimizes revenue over the space of all feasible interim rules has proven to be a useful endeavor when computing optimal and approximately optimal mechanisms~\cite{CDW_1, CDW_2, CDW_3, CDW_4}, as well as for deriving upper bounds (via duality) to the revenue optimal mechanism~\cite{cai2019duality}. While the number of variables in this program (corresponding to interim rules) is polynomial, the number of constraints needed to ensure that an interim rule is feasible (i.e., that there exists a mechanism $\M$ that induces it) is typically exponential (even for, e.g., the simple case of selling a single item to $n$ agents).
The starting point of our approach is to consider a relaxation of these feasibility constraints, resulting in interim rules are \emph{feasible in expectation}.




Given (optimal or approximately optimal) interim rules that are feasible in expectation the first natural step for rounding to an actual mechanism is to use a CRS/OCRS. Contention resolution schemes, or CRSs, were defined by Chekuri et al.~\cite{chekuri2014submodular} as a tool for rounding fractional solutions in (submodular) optimization problems. In this framework, there is a finite ground set of elements $N = \{ e_1, \dots, e_k \}$, a downward-closed family $\Fcal$ of subsets of $N$, and a fractionally feasible point $x^* \in [0,1]^k$. The main idea is to obtain a (possibly infeasible) random set $R(x^*)$ from $x^*$, by treating $x^*$ as a product distribution: element $i$ is included in $R(x^*)$ with probability $x^*_i$. Given $R(x^*)$, a $c$-selectable CRS selects a set $I \subseteq R(x^*)$ that is feasible (i.e., $I \in \Fcal$), in a way that each element is selected with probability at least $c$ if it is in $R(x^*)$. A refined definition, $(b,c)$-selectable CRSs, for $b \in [0,1]$, extends the concept of $c$-selectable CRSs (which are simply $(1,c)$-selectable) and provides the same guarantee per element when given as input a set $R(b \, x^*)$. Feldman et al.~\cite{feldman2021online} extended the CRS framework to online settings and defined OCRSs.

Back to our problem, a first blueprint for a mechanism, given interim rules that are feasible in expectation, would be to elicit reports $r_1, \dots, r_n$ from the agents, and then construct a set of active elements (the input to the CRS/OCRS) according to $\pi_{i,j}(r_i)$, for every agent $i$ and item $j$. A technical complication is that CRSs/OCRSs receive as input elements that become active \emph{independently}. In our blueprint, the event that element ``$(i,j,r_i)$'' becomes active is correlated with the event that element ``$(i,j',r_i)$'' becomes active. To bypass this obstacle, we define variants of CRSs and OCRSs which we call \emph{two-level contention resolution schemes}, or tCRS, and \emph{two-level online contention resolution schemes}, or tOCRS. Intuitively, a tCRS/tOCRS receives an $n$ by $m$ matrix of elements, such that elements in the same row are independent, conditioned on the value of a row-specific random variable (and these row-specific random variables are independent).


\paragraph{Our framework.}
Informally, our overall framework takes as input feasible in expectation interim rules and a tCRS/tOCRS and outputs a mechanism; for the case of tOCRS the mechanism is sequential, i.e., it sequentially approaches each agent $i$, elicits a report $r_i$, and decides on the outcome of agent $i$ (her allocation and payment) before proceeding to the next agent.
We require that valuations are independent across agents, but the value of agent $i$ for item $j$ can be correlated with her value for item $j'$.
Our mechanisms are Bayesian incentive compatible (BIC) and Bayesian individually rational (BIR). Given an $\alpha \geq 1$ approximately optimal interim form, and a $(b,c)$-selectable tCRS/tOCRS, our mechanism is $\frac{\alpha}{b \, c}$ approximately optimal (\Cref{thm: from t to mechanisms}). See~\Cref{fig:example}.

\begin{figure}[hbt]
  \centering
  \includegraphics[width=0.8\linewidth]{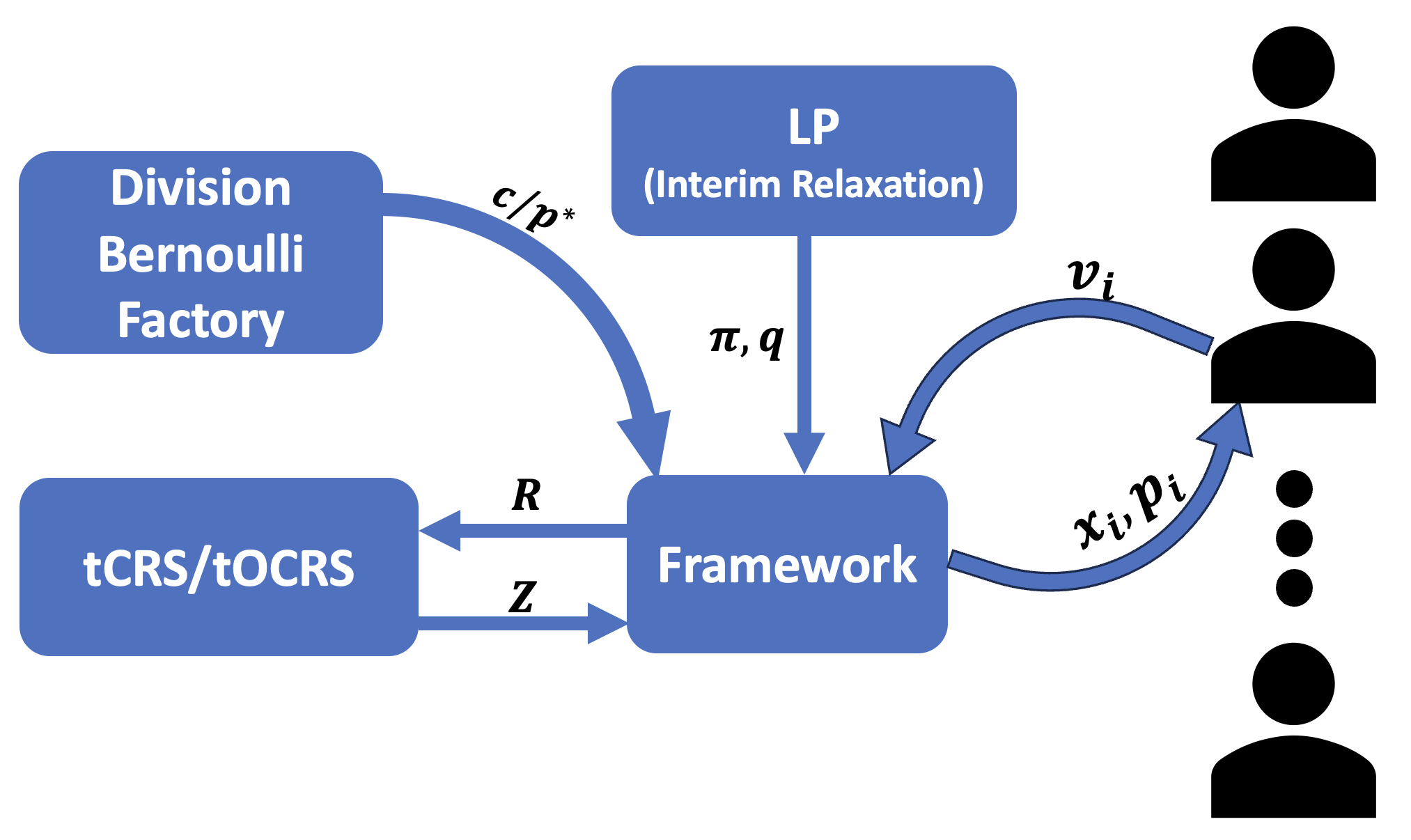} 
\caption{Given an interim form $(\pi,q)$, a feasible solution for~\eqref{lp}, our framework (sequentially in the case of tOCRSs) elicits valuations from the agents. Given the report of agent $i$, $v_i$, it executes the tCRS/tOCRS on a set of active elements $R$, which returns a set $Z$. The allocation of agent $i$ is constructed from $Z$; when given only black-box access to the tCRS/tOCRS this construction can be done efficiently via a Bernoulli factory for division.}
  \label{fig:example}
\end{figure}

Towards constructing tCRSs/tOCRSs, we first give a general reduction for constructing a tCRS/tOCRS, that uses CRSs/OCRSs as a black-box (\Cref{theorem:vertical-horizontal}). Informally, if CRSs/OCRSs for certain feasibility sets $\Fcal_1, \dots, \Fcal_k$ exist, we can provide a tCRS/tOCRS where any combination of the $\Fcal_i$s is satisfied for elements in the same row or column. Combining with known CRS/OCRS results we get tCRSs/tOCRSs for various settings of interest.
Next, we give $1/10$-selectable tOCRS for Knapsack constraints and a $1/9$-selectable tOCRS for Multi-Choice Knapsack constraints (Theorems~\ref{theorem:knapsack} and~\ref{theorem:knapsack unit demand}). Notably, our tOCRS for knapsack implies a $1/10$-selectable OCRS for knapsack, which is better than the $0.085$-selectable OCRS given by Feldman et al.~\cite{feldman2021online}, but not as good as the state-of-the-art $1/(3+e^{-2})$-selectable ($\approx 0.319$-selectable) OCRS of Jiang et al.~\cite{jiang2022tight}.

\paragraph{Applications.} Plugging the aforementioned tCRSs/tOCRSs into our framework gives numerous interesting applications. As a first application, consider the problem of auctioning off $m$ items to $n$ agents with additive preferences, such that the set of agents that each item $j$ is allocated to must be an independent set of a matroid, and the items allocated to each agent must be an independent set of a matroid. Our results imply that there exists a sequential, BIC and BIR mechanism that guarantees a $\frac{2e}{e-1} \approx 3.16$ approximation to the optimal revenue (\Cref{application: matroid constraint}). 
The previously best-known approximation possible by a sequential mechanism was $70$ (for the special case where every item can be allocated at most once, and where items' values are independent), due to Cai and Zhao~\cite{cai2017subadditive}, whose mechanism has additional desirable simplicity properties (that we do not guarantee here, and are in fact impossible to guarantee for correlated items~\cite{briest2015pricing,hart2019selling}). We note that, to get an end-to-end sequential mechanism efficiently, one would need to efficiently compute an optimal, feasible in expectation interim form. One way to compute such rules is to solve a simple linear program (\ref{lp}), which has polynomially many constraints for many natural choices of matroids (e.g., uniform matroids). Additionally, if one is given only black-box query access to the required (for this application) tOCRS, it is still possible to implement a sequential mechanism efficiently, via Bernoulli factories, without any loss in the BIC guarantees, but by incurring a small additional cost in approximation. See~\Cref{subsec: implementation} for details.

As a next application, consider the problem of auctioning off $m$ items to $n$ agents with additive preferences, where each item $j \in [m]$ has some weight $k_j$ and the total weight of items sold cannot exceed $K$. Our results imply that there exists an efficiently computable, sequential, BIC and BIR mechanism that guarantees a $10$ approximation to the optimal revenue. Additionally, if each agent $i$ can get at most one item, there exists an efficiently computable, sequential, BIC and BIR mechanism that guarantees a $9$ approximation to the optimal revenue (\Cref{application: Knapsack}). Finally, consider the problem of auctioning off $m$ items to $n$ agents with \emph{arbitrary} valuation functions, where each item $j \in [m]$ has some weight $k_j$ and the total weight of items sold cannot exceed $K$. Then, our results imply that there exists a (computationally inefficient) sequential, BIC and BIR mechanism that guarantees a $9$ approximation to the optimal revenue (\Cref{application: arbitrary}). 
\vspace{-3mm}
\paragraph{Extensions.}
Our framework can be easily extended to other mechanism design problems, beyond auctioning off items to agents. In~\Cref{sec: procurement} we give an extension to \emph{procurement auctions}, where a value-maximizing buyer is interested in buying services from strategic sellers, subject to a budget constraint. In this case, we show how, given an OCRS for Stochastic Knapsack,\footnote{See~\Cref{sec: procurement} for definitions.} it is possible to design an approximately optimal sequential multi-parameter procurement auction. Combining with a result of Jiang et al.~\cite{jiang2022tight}, we then have a $(3+e^{-2})$-approximately optimal sequential procurement auction (\Cref{application: procurement}). As an aside, we also give an OCRS for the stochastic knapsack setting that might be of independent interest (\Cref{theorem:stochastic knapsack}). Feldman et. al.~\cite{feldman2021online} give a greedy and monotone\footnote{An OCRS is greedy if it fixes a downward-closed family of feasible sets before the (online) process starts, and greedily accepts any active element $e$ that will not violate feasibility if included. An OCRS $\mu$ is \emph{monotone} if for all $e \in A \subseteq B$, the probability that $\mu$ selects $e$ when $A$ is the set of active elements is at most the probability that $\mu$ selects $e$ when $B$ is the set of active elements. These properties are important for applications in submodular optimization~\cite{chekuri2014submodular}.} $(3/2- \sqrt{2})$-selectable OCRS $(\approx 0.0858)$ for Knapsack, while Jiang et. al.~\cite{jiang2022tight} give a $\frac{1}{3+e^{-2}}$-selectable OCRS $(\approx 0.319)$ for Stochastic Knapsack (this OCRS induces a non-greedy and non-monotone OCRS for Knapsack). We give $c$-selectable OCRS for Stochastic Knapsack (that induces a greedy and monotone OCRS for Knapsack), where $c = \max\{\frac{1-k^*}{2-k^*}, 1/6\}$, and $k^*$ is a parameter that depends on the maximum possible weight (in the support of the distributions from which the stochastic weights are drawn); our OCRS is therefore always better than the OCRS of Feldman et. al., and better than the OCRS of Jiang et. al. when $k^*$ is small.

\vspace{-3mm}
\paragraph{Roadmap.}

We provide background on CRSs/OCRSs and Bernoulli factories, and describe our model and our extension to two-level CRSs/OCRSs in~\Cref{sec:prelims}.
In~\Cref{sec: mech from tCRS} we show how to construct mechanisms given tCRSs and tOCRSs. In~\Cref{sec:constructing tCRS/tOCRS} we show how to construct tCRSs and tOCRSs. In~\Cref{sec:applications} we put together the results from the two previous sections to give end-to-end mechanisms. In~\Cref{sec: procurement} we extend our framework to procurement auctions. Further related work is discussed in~\Cref{sec: related work}.

\section{Preliminaries}\label{sec:prelims}

We consider the problem of a seller with $m$ indivisible, heterogeneous items for sale to $n$ strategic agents.
Each agent $i$ has a private valuation vector $v_i$ that is drawn independently from an $m$-dimensional distribution $\Dcal_i$ (that is known to the seller). We write $\Vcal_i = supp(\Dcal_i)$ for the set of possible valuations for agent $i$.
Agent $i$ has a value $v_{i,j}$ for item $j$. We write $\Dcal_{i,j}$ for the marginal distribution for item $j$, noting that $\Dcal_{i,j}$ is not necessarily independent of $\Dcal_{i,j'}$. We assume that agents have \emph{additive preferences}, i.e.,
the value of agent $i$ with valuation $v_i$ for a subset of items $S \subseteq [m]$ is $\sum_{j \in S} v_{i,j}$. agents are \emph{quasi-linear}: the utility of an agent is her value minus her payment. An (integral) allocation $x \in \{ 0, 1 \}^{n \cdot m}$ indicates which item was received by which agent, i.e., $x_{i,j} \in \{ 0, 1 \}$ is the indicator for whether agent $i$ received item $j$. There are constraints on the set of feasible allocations represented by a set $\Fcal \subseteq \{ 0, 1 \}^{n \cdot m}$; that is, an allocation $x$ is feasible if $x \in \Fcal$ (therefore, one can equivalently think of the agents as constrained additive). Let $P_{\Fcal}$ be the convex hull of all characteristic vectors of $\Fcal$, i.e. $P_{\Fcal} = conv\{\mathbf{1_{F}}: F \in \Fcal\}$. We write $P_{\Fcal}^i$ for the polytope that corresponds to agent $i$, i.e., the polytope $P_{\Fcal}$ when we only consider the $m$ dimensions that correspond to the allocation of agent $i$. 


\subsection{Mechanism Design Preliminaries}

A mechanism $\M$ takes as input a reported valuation from each agent and selects a (possibly random) allocation in $\Fcal$, and payments to charge the agents. An agent's objective is to maximize her expected utility. A mechanism $\M$ is \emph{Bayesian Incentive Compatible (BIC)} if every agent $i \in [n]$ maximizes her expected utility by reporting her true valuation $v_i$, assuming other agents do so as well, where this expectation is over the randomness of other agents' valuations, as well as the randomness of the mechanism. A mechanism is \emph{Bayesian Individually Rational} (BIR) if every agent $i \in [n]$ has non-negative expected utility when reporting her true valuation (assuming other agents do so as well). The (expected) revenue of a BIC mechanism is the expected sum of payments made when agents draw their valuations from $\Dcal$ (and report their true valuations to the mechanism). We say that a mechanism is BIC-IR if it is both BIC and BIR.

The \emph{optimal mechanism} for a given distribution $\Dcal$, whose revenue is denoted by $\Rev(\Dcal)$, maximizes expected revenue over all BIC-IR mechanisms. For a given mechanism $\M$, we slightly abuse notation and write $\Rev^{\M}(\Dcal)$ to denote its revenue under a distribution $\Dcal$. A mechanism guarantees an $\alpha$ approximation to the optimal revenue if $\alpha \Rev^{\M}(\Dcal) \geq \Rev(\Dcal)$. Finally, we say that a mechanism is \emph{sequential} if it sequentially approaches agent $i$, elicits a report $r_i$, and allocates items to $i$ before proceeding to the next agent.

\paragraph{Interim allocations and payments.}
The \emph{interim allocation} of a mechanism $\M$, $\pi^{\M}$, indicates, for each agent $i$ and item $j$ the probability $\pi^{\M}_{i,j}(r_i)$ that agent $i$ receives item $j$ when she reports valuation $r_i$ (over the randomness in $\M$ and the randomness in other agents' reported valuations $v_{-i}$, drawn from $\Dcal_{-i}$). The \emph{interim payment} of agent $i$, $q^{\M}_{i}(r_i)$, is the expected payment she makes when she reports valuation $r_i$ (again, over the randomness in $\M$ and the randomness in other agents' reported valuations).
It is easy to see that the expected utility of agent $i$ with valuation $v_i$ when reporting $r_i$ to a mechanism $\M$, is $\sum_{j \in [m]} v_{i,j} \pi^{\M}_{i,j}(r_i) - q^{\M}_{i}(r_i)$. We will drop the superscript $\M$ when the mechanism is clear from the context.


Given interim allocations and payments, it is not a straightforward task to determine whether they are \emph{ex-post feasible}, i.e., whether there exists a mechanism $\M$ that induces the exact probabilities promised by the interim allocations.\footnote{Doing this task efficiently is at the core of the framework of Cai et. al.~\cite{CDW_1, CDW_2, CDW_3, CDW_4} for computing approximately optimal mechanisms.} However, it is typically straightforward to find interim allocations that are \emph{feasible in expectation}.
\begin{definition}[Feasibility in expectation]
An interim allocation rule $\pi$ is feasible in expectation if (i) $\forall i\in [n], v_i \in \Vcal_i$, $\pi_i(v_i)\in P_{\Fcal}^i$, and (ii) $\forall i\in [n], j \in [n], \sum_{v_i \in \Vcal_i} \Pr[v_i] \cdot \pi_{i,j}(v_i) \in P_{\Fcal}$.
\end{definition}

We say that an interim allocation, payment pair $(\pi, q)$ is BIC if $\forall i\in [n], v_i, v'_i \in \Vcal_i$ it holds that 
\[
\sum_{j \in [m]} v_{i,j}\pi_{i,j}(v_i) - q_i(v_i) \ge \sum_{j \in [m]} v_{i,j}\pi_{i,j}(v'_i) - q_i(v'_i).
\]
An interim allocation, payment pair $(\pi, q)$ is BIR if $\forall i\in [n], v_i \in \Vcal_i$, $\sum_{j \in [m]} v_{i,j}\pi_{i,j}(v_i) - q_i(v_i) \geq 0$.
An interim allocation, payment pair $(\pi, q)$ is BIC-IR if it is both BIC and BIR.
Finally, an interim allocation, payment pair $(\pi, q)$ guarantees an $\alpha$-approximation to the optimal revenue if $\alpha \left( \sum_{i \in [n]} \sum_{v_i \in \Vcal_i} \Pr[v_i]  \cdot q_{i}(v_i) \right) \geq \Rev(\Dcal)$.

\subsection{CRS/OCRS and tCRS/tOCRS Preliminaries}




Consider a finite ground set $N = \{e_1, \cdots, e_k\}$ and a family of feasible subsets $\Fcal \subseteq 2^N$. Let $P_{\Fcal} = conv(\{\mathbf{1}_{I}|I \in \Fcal\})$ be the convex hull of all characteristic vectors of feasible sets. Let $x \in P_{\Fcal}$, and let $R(x) \subseteq N$ be a random set obtained by including each element $i \in N$ independently with probability $x_i$. The set $R(x)$ is feasible in expectation (with respect to $\Fcal$) but not necessarily ex-post feasible.
Given $R(x)$, a contention resolution scheme (CRS) selects a subset $I \subseteq R(x)$ such that $ I \in \Fcal$. If elements of $R(x)$ are given in an online manner, the corresponding scheme is called an online contention resolution scheme (OCRS). To avoid trivial solutions (e.g., $I = \emptyset$), we would additionally like to have the property that each element $i \in N$ appears in $I$ with probability at least $c x_i$ for some $c$. We call such schemes $c$-selectable. Some schemes only work if elements come from $R(b \,x)$; such schemes are called $(b,c)$-selectable. Formally:

\begin{definition}[Contention Resolution Scheme (CRS)~\cite{chekuri2014submodular}]
    Let $b,c \in [0,1]$. For every $x \in P_{\Fcal}$, let $R(b 
    \, x)$ be a random subset of active elements, where element $i \in N$ is active with probability $b \, x_i$, independently. A $(b,c)$-selectable Contention Resolution scheme (CRS) $\mu$ for $P_{\Fcal}$ is a (possibly randomized) procedure that, given a set of active elements $R(b 
    \, x)$ returns a set $\mu(R(b \, x)) = I \subseteq R(b \, x)$, such that (i) $I \in \Fcal$, and (ii) $\Pr{\left[i \in I | i \in R(b \, x) \right]} \ge c, \forall i \in N$.
\end{definition}

\begin{definition}[Online Contention Resolution Scheme(OCRS)~\cite{feldman2021online}]
    Let $b,c \in [0,1]$. For every $x \in P_{\Fcal}$, let $R(b 
    \, x)$ be a random subset of active elements, where element $i \in N$ is active with probability $b \, x_i$, independently. A $(b,c)$-selectable Online Contention Resolution scheme (OCRS) $\mu$ for $P_{\Fcal}$ is a (possibly randomized) online procedure that, given active elements one by one, decides whether to select an active element irrevocably before the next element is revealed. The OCRS $\mu$ returns a set $I \subseteq R(b \, x)$, such that (i) $I \in \Fcal$, and (ii) $\Pr{\left[i \in I | i \in R(b \, x) \right]} \ge c, \forall i \in N$.
\end{definition}





We introduce a variant of the previous CRS/OCRS model that allows for dependencies between the activation of different elements. This slightly changes the setup, as well as the definition of a ``scheme.'' Consider the ground set $N = \{ e_{i,j} \}_{i \in [n], j \in [m]}$, where $|N| = n \, m$, and a family of feasible subsets $\Fcal \subseteq 2^N$. Let $P_{\Fcal} = conv(\{\mathbf{1}_{I}|I \in \Fcal\})$ be the convex hull of all characteristic vectors of feasible sets. 
Let $P_{\Fcal}^i$ be the restriction of $P_{\Fcal}$ to the $m$ dimensions that correspond to elements $(i,j)$, $j \in [m]$. Elements will become active in a certain, dependent way, as induced by a \emph{two-level stochastic process} $(\Dcal,x)$, defined as follows:

\begin{definition}[Two-Level stochastic process]
    We say that $(\Dcal, x)$ is a \emph{two-level stochastic process} over $\{ 0, 1 \}^{n \, m}$, where $\Dcal = \times_{i=1}^n \Dcal_i$ is a product distribution and $x \in [0,1]^{m \sum_{i=1}^n |\Vcal_i| }$,
    if it is induced by the following procedure: (i) we first sample $d_i$ from $\Dcal_i$, independently, and (ii) for each $(i,j) \in [n]\times[m]$, element $e_{i,j}$ becomes active with probability $x_{i,j}(d_i)$, independently.
\end{definition}

For the case of CRSs/OCRSs, since elements became active independently, (expected) feasibility for a vector/product distribution $x$ boiled down to $x$ being fractionally feasible with respect to $\Fcal$, i.e. $x \in P_{\Fcal}$. Here, since elements are active in a dependent way, our notion of feasibility needs to be refined:
\begin{definition}[Feasibility]
    Let $\Fcal \subseteq 2^{[n]\times [m]}$ be a feasibility set and $P_{\Fcal}$ be its relaxation. We say that a two-level stochastic process $(\Dcal, x)$ is feasible with respect to $\Fcal$ if:
    \begin{enumerate}
        \item For each $i \in [n]$ and each $d_i \in supp(\Dcal_i)$, $\left( x_{i,1}(d_i), \cdots, x_{i,m}(d_i) \right) \in P^i_{\Fcal}$.
        \item $w \in P_{\Fcal}$, where $w_{i,j} = \sum_{d_i \in supp(\Dcal_i)} \Pr\left[ \Dcal_i = d_i \right] x_{i,j}(d_i)$.  
    \end{enumerate}
\end{definition}

Let $(\Dcal, x)$ be a two-level stochastic process that is feasible with respect to $\Fcal$, and let $R(\Dcal, x) \subseteq N$ be a random set of elements obtained by sampling from $(\Dcal, x)$. Our goal is to select a subset $I \subseteq R(\Dcal, x)$ (possibly online) such that $I \in \Fcal$ and the probability that an active element is selected is lower bounded by a constant $c$. Formally, our definitions of two-level Contention Resolution Schemes (tCRSs) and two-level Online Contention Resolution Schemes (tOCRSs) are as follows.

\begin{definition}[Two-level CRS (tCRS)]
    Let $b, c \in [0,1]$.
    Let $(\Dcal, x)$ be a two-level stochastic process that is feasible with respect to $\Fcal$, and let $R(\Dcal, b \, x) \subseteq N$ be a random set of elements obtained by sampling from $(\Dcal, b \, x)$.
    A $(b,c)$-selectable \emph{two-level CRS} (tCRS) $\mu$ for $\Fcal$ is a (possibly randomized) procedure that, given a set of active elements $R(\Dcal, b \, x) \subseteq N$, returns a set $\mu( R(\Dcal, b \, x) ) = I \subseteq R(\Dcal, b \, x)$, such that (i) $I \in \Fcal$, and (ii) $\Pr{\left[i \in I | i \in R(\Dcal, b \, x) \right]} \ge c, \forall i \in N$.
\end{definition}


\begin{definition}[Two-level OCRS (tOCRS)]
    Let $b, c \in [0,1]$.
    Let $(\Dcal, x)$ be a two-level stochastic process that is feasible with respect to $\Fcal$, and let $R(\Dcal, b \, x) \subseteq N$ be a random set of elements obtained by sampling from $(\Dcal, b \, x)$.
    Elements of $R(\Dcal, b \, x)$ appear online, in batches of size $m$: the process selects some $i \in [n]$ and reveals all elements $\{ e_{i,j} \}_{j \in [m]}$, before selecting a new $i' \in [n]$. A $(b,c)$-selectable \emph{two-level OCRS} (tOCRS) $\mu$ for $\Fcal$ is a (possibly randomized) online procedure that, given active elements that satisfy the aforementioned ordering, decides whether to select an active element irrevocably before the next element is revealed, i.e., returns a set $I \subseteq R(\Dcal, b \, x)$, such that (i) $I \in \Fcal$, and (ii) $\Pr{\left[i \in I | i \in R(\Dcal, b \, x) \right]} \ge c, \forall i \in N$.
\end{definition}


It is easy to see that the existence of a $(b,c)$-selectable tCRS (resp. tOCRS) implies the existence of a $(b,c)$-selectable CRS (resp. OCRS), by a simple simulation argument (for $m=1$). 

\subsection{Bernoulli Factory Preliminaries}

Bernoulli factories were introduced by Kean et. al.~\cite{keane1994bernoulli}, where they are defined as follows.

\begin{definition}[Bernoulli Factory]
  Given a function $f: (0,1) \rightarrow (0,1)$, a Bernoulli factory for $f$ outputs a sample of a Bernoulli variable with bias $f(p)$ (i.e. an $f(p)$-coin), given black-box access to independent samples of a Bernoulli distribution with bias $p \in (0,1)$ (i.e. a $p$-coin).  
\end{definition}

As an illustrative example, imagine that we are given a $p$-coin, a coin that outputs $1$ with probability $p$ and $0$ otherwise. Our goal is to create a new coin that outputs $1$ with probability $f(p)=p^2$. The complication here is that we do not know the value of $p$. $f(p) = p^2$ can be implemented as follows: flip the $p$-coin twice. If both are $1$ then output $1$ (otherwise output $0$). We include additional examples of Bernoulli factories in~\Cref{app: bernoulli}. 
Bernoulli factories have recently been used in mechanism design in the context of black-box reductions~\cite{dughmi2021bernoulli,cai2021efficient}. In this paper, we make use of a Bernoulli factory for division: given one $p_0$-coin and one $p_1$-coin, implement $f(p_0, p_1) = p_0/p_1$ for $p_1-p_0 \ge \delta$. This problem was considered by Nacu et. al.~\cite{nacu2005fast} but their construction is rather involved. Instead, consider~\Cref{algorithm:bernoulli}, the Bernoulli Division factory of Morina~\cite{morina2021bernoulli}. 

\begin{algorithm}[ht]
\caption{Bernoulli Division \cite{morina2021bernoulli}}\label{algorithm:bernoulli}
 \While{true}{
    $X \sim Bern[1/2]$.\label{alg:goto}\\
    \If{X = 0}{
        $W \sim Bern[p_0]$.\\
        \If{$W=1$}{return $1$}
    }
    \Else{ 
        $W \sim Bern[p_1-p_0]$.\\
        \If{$W=1$}
        { return $0$}
    }
 }
\end{algorithm}

\begin{lemma}[\cite{morina2021bernoulli}]\label{lem: bernoulli for division}
Given a $p_0$-coin and a $p_1$-coin, assume $p_1 - p_0 \geq \delta$, and let $N$ be the number of tosses required. Then, \Cref{algorithm:bernoulli} is a Bernoulli factory for $(p_0/p_1)$ which satisfies $\Ex{N} \leq \frac{22.12}{p_1}(1+\delta^{-1})$.
\end{lemma}

 Although the process and its correctness are fully described by Morina~\cite{morina2021bernoulli}, the end-to-end expected number of tosses is not explicitly calculated. For completeness, we show these calculations (and also argue that~\Cref{algorithm:bernoulli} is a Bernoulli factory for $(p_0/p_1)$) in~\Cref{app: bernoulli}.

\section{Mechanisms from two-level CRSs and two-level OCRSs}\label{sec: mech from tCRS}

We give two frameworks for constructing mechanisms. The input to a framework is (i) a BIC-IR interim allocation, payment rule pair $(\pi,q)$ that is feasible in expectation and is an $\alpha$ approximation to the optimal revenue, (ii) a $(b,c)$-selectable tCRS/tOCRS for $\Fcal$, and (iii) a parameter $\epsilon \geq 0$. Our frameworks produce BIC-IR, $\frac{\alpha}{b \, (c - \epsilon)}$ approximately optimal (and sequential, for tOCRSs) mechanisms for agents with constrained (with respect to $\Fcal$) additive valuations.

Given a tCRS, our framework,~\Cref{algo: tCRS to mech}, works as follows. First, each agent $i$ reports $r^*_i \in \Vcal_i$ and pays $b \, ( c - \epsilon ) \cdot q_i (r^*_i)$. We then construct a set $R = \{ x_{i,j} \}_{i \in [n], j \in [m]}$ of $n \cdot m$ elements, one for every agent, item pair, where $x_{i,j} = 1$ (i.e., element $(i,j)$ is active) with probability $b \cdot \pi_{i,j}(r^*_i)$. We query the tCRS on input $R$, to get back a subset  $Z$ of selected elements (which is in $\Fcal$, by definition). We flip an additional coin to decide whether to keep an element $(i,j)$, i.e. whether we should allocate item $j$ to agent $i$. This last coin essentially balances the randomness of the tCRS and ensures that the probability that agent $i$ gets item $j$ when they report $r^*_i$ is exactly $b (c-\epsilon) \pi_{i,j}(r^*_i)$, which, combined with the chosen payment, ensures that BIC and BIR are satisfied. 

Given a tOCRS, our framework,~\Cref{algo: tOCRS to mech}, works similarly. We approach each agent $i$ sequentially. Agent $i$ reports $r^*_i \in \Vcal_i$ and pays $b \, ( c - \epsilon ) \cdot q_i (r^*_i)$. We consider each item $j$ sequentially. We make element $(i,j)$ active with probability $b \cdot \pi_{i,j}(r^*_i)$, and then ask the tOCRS if this element should be selected (assuming it was active). If the tOCRS selects the element $(i,j)$, we again flip an additional coin to decide whether agent $i$ should get item $j$.

\begin{algorithm}[!ht]
\DontPrintSemicolon
  \KwInput{allocation, payment rule pair $(\pi,q)$, $(b,c)$-selectable tCRS $\mu$, parameter $\epsilon \geq 0$.}
   \hrulealg
  \For{each $i \in [n]$}{
        Agent $i$ reports $r^*_i \in \mathcal{V}_i$ and pays $b \, (c - \epsilon) \, q_i(r^*_i)$.
        }
  Construct $R = \{ x_{i,j} \}_{i \in [n], j \in [m]}$, where $x_{i,j} \gets 1$ with probability $b \, \pi_{i,j}(r^*_i)$, and $x_{i,j} \gets 0$ otherwise.\\
  $Z = \{ Z_{i,j} \}_{i \in [n], j \in [m]} \gets \mu (R)$. \tcp*{$Z \subseteq R$ is the set of elements that the tCRS picks}
  \For{each element $(i, j) \in [n] \times [m]$ with $Z_{i,j} = 1$}{
        Allocate item $j$ to agent $i$ with probability $\frac{c - \epsilon}{p^*_{i,j}(r^*_i)}$, where $p^*_{i,j}(r^*_i)$ is the probability that $\mu$ selects $(i,j)$ conditioned on $i$'s report $r^*_i$ and $x_{i,j}$ being equal to $1$.\label{eq: flip coin in tCRS}
  }
\caption{Our framework for tCRSs}\label{algo: tCRS to mech}
\end{algorithm}

\begin{algorithm}[!ht]
  \DontPrintSemicolon
  \KwInput{allocation, payment rule pair $(\pi,q)$, $(b,c)$-selectable tOCRS $\mu$, parameter $\epsilon \geq 0$.}
   \hrulealg
  \For{each agent $i \in [n]$}{
        Agent $i$ reports $r^*_i \in \mathcal{V}_i$ and pays $b \, (c - \epsilon) \, q_i(r^*_i)$.\\
        \For{each item $j \in [m]$}{
            $R_{i,j} \gets 1$ with probability $b \, \pi_{i,j}(r^*_i)$.\\
            $Z_{i,j} \gets \mu(R_{i,j})$. \tcp*{$Z_{i,j} \leq R_{i,j}$ indicates if the tOCRS selects element $(i,j)$ when active}
            \If{$Z_{i,j} = 1$}{
                Allocate item $j$ to agent $i$ with probability $\frac{c - \epsilon}{p^*_{i,j}(r^*_i)}$, where $p^*_{i,j}(r^*_i)$ is the probability that $\mu$ selects $(i,j)$ conditioned on $i$'s report $r^*_i$ and $R_{i,j}$ being equal to $1$. \label{eq: flip coin tOCRS}
            }
        }    
    }
\caption{Our framework for tOCRSs}\label{algo: tOCRS to mech}
\end{algorithm}

\vspace{-1mm}
\begin{theorem}\label{thm: from t to mechanisms}
    Given (i) a BIC-IR interim allocation, payment rule pair $(\pi,q)$ that is feasible in expectation and is an $\alpha \geq 1$ approximation to the optimal revenue (ii) a $(b,c)$-selectable tCRS (resp. tOCRS) for $\mathcal{F}$,  and (iii) a parameter $\epsilon \geq 0$,~\Cref{algo: tCRS to mech} and~\Cref{algo: tOCRS to mech} give BIC-IR, $\frac{\alpha}{b \, (c - \epsilon)}$-approximately optimal (and sequential for the case of~\Cref{algo: tOCRS to mech}) mechanisms for $\mathcal{F}$. If we only have query access to the tCRS/tOCRS, our mechanisms can be implemented using a $O(poly( \sum_i |\Vcal_i|, m, \frac{1}{\epsilon}))$ number of queries in expectation.
\end{theorem}

\begin{proof}
First, we argue that our frameworks output allocations in $\Fcal$.
By definition, and assuming truthful reports, the interim allocation $\pi$ defines a feasible (for $\Fcal$) two-level stochastic process, where we first sample $v_i$s independently, and then $\pi_{i,j}(v_i)$. Let $x \in R(b \, \pi)$ be the input to the $(b,c)$-selectable tCRS (resp. for tOCRS); by definition, the set of selected elements $Z$ satisfies $Z \in \Fcal$. We allocate a subset of $Z$, thus our allocations are feasible, since $\Fcal$ is downward closed. Furthermore, for any element $z_{i,j}$, $\Pr[ z_{i,j} = 1 | x_{i,j} = 1 ] = p^*_{i,j}(v_i) \geq c$, and thus $\frac{c - \epsilon}{p^*_{i,j}(v_i)}$ is a probability.
Second, we argue that our mechanisms are BIC. From the perspective of agent $i$, a report $r_i \in \Vcal_i$ costs  $b \, (c - \epsilon) \cdot q_i(r_i)$ and allocates item $j$ with probability $b \, \pi_{i,j}(r_i) \cdot p^*_{i,j}(r_i) \cdot \frac{c - \epsilon}{p^*_{i,j}(r_i)} = b \, (c - \epsilon) \cdot \pi_{i,j}(r_i)$, and therefore translates to an expected utility of costs $b \, (c - \epsilon) \sum_{j \in [m]} v_{i,j} \pi_{i,j}(r_i) - b \, (c - \epsilon) \cdot q_i(r_i)$; since $(\pi, q)$ is BIC, so is the mechanism we output. Near-identical arguments imply the BIR guarantee and revenue guarantees.

When given only black-box access to a tCRS/tOCRS, it is not immediately clear how one can flip a coin with probability exactly $(c - \epsilon)/p^*_{i,j}(v_i)$ (efficiently or otherwise), as needed in step~\eqref{eq: flip coin in tCRS} of~\Cref{algo: tCRS to mech} and step~\eqref{eq: flip coin tOCRS} of~\Cref{algo: tOCRS to mech}. Using a Bernoulli factory for division (such as the result of~\cite{morina2021bernoulli} discussed in~\Cref{sec:prelims}), this step can be implemented 
using $\frac{22.12}{p^*_{i,j}(v_i)}(1+\epsilon^{-1}) \in O( 1 / \epsilon )$ calls in expectation (\Cref{lem: bernoulli for division}), assuming that $c \leq p^*_{i,j}(v_i)$ is a constant. Overall, we have $O(poly( \sum_i |\Vcal_i|, m, \frac{1}{\epsilon}))$ queries in expectation for the entire execution of a framework; we discuss implementation details in~\Cref{subsec: implementation}.
\end{proof}


\vspace{-6mm}
\subsection{Implementation Considerations}\label{subsec: implementation}

Here, we highlight some implementation details for our framework. First, we give a simple LP that computes optimal ($\alpha = 1$), feasible in expectation, BIC-IR interim allocation and payment rules $(\pi,q)$. Second, we flesh out implementation details of step~\eqref{eq: flip coin in tCRS} of~\Cref{algo: tCRS to mech} and step~\eqref{eq: flip coin tOCRS} of~\Cref{algo: tOCRS to mech}, flipping a coin with probability $(c - \epsilon)/p^*_{i,j}(v_i)$, when given only black-box access to a tCRS/tOCRS. Note that this is not a straightforward task, since $p^*_{i,j}(v_i)$ might not be known to the algorithm (otherwise, this task is indeed trivial), and approximating this probability (e.g., via repeated simulations), instead of computing it exactly, results in a violation of the BIC constraint.

\paragraph{Finding feasible in expectation, optimal interim rules.}
Consider the following linear program,~\eqref{lp}, which computes an interim relaxation of the revenue optimal BIC-IR mechanism.

\begin{equation}
\label{lp}
\begin{array}{ll@{}ll}
\text{maximize}  & \displaystyle\sum\limits_{i \in [n]} \displaystyle\sum\limits_{v_i \in \Vcal_i} \Pr[v_i] q_i(v_i)\\ \\
\text{s.t.}& \sum_{j \in [m]} v_{i,j}\pi_{i,j}(v_i) - q_i(v_i) \ge \sum_{j \in [m]} v_{i,j}\pi_{i,j}(v'_i) - q_i(v'_i) && \forall i\in [n], v_i, v'_i \in \Vcal_i\\ \\
&  \sum_{j \in [m]} v_{i,j}\pi_{i,j}(v_i) - q_i(v_i) \ge 0 && \forall i\in [n], v_i \in \Vcal_i  \\ \\
& \pi_i(v_i)\in P_{\Fcal}^i && \forall i\in [n], v_i \in \Vcal_i \\ \\
&  \left[\sum_{v_i \in \Vcal_i} \Pr[v_i] \cdot \pi_{i,j}(v_i) \right]_{(i,j) \in [n] \times [m]} \in P_{\Fcal}
\end{array}
\tag{LP1}
\end{equation}

This linear program has $O(m \sum_{i \in [n]} | \Vcal_i |)$ variables and $O(\sum_{i \in [n]} | \Vcal_i |^2)$ constraints, excluding the constraints necessary to represent $P_{\Fcal}^i$ and $P_{\Fcal}$. Therefore, the overall computational complexity of solving this LP depends on whether $P_{\Fcal}$ and $P_{\Fcal}^i$ have an efficient representation.

We note that, in a series of papers, Cai et. al.~\cite{CDW_1, CDW_2, CDW_3, CDW_4} propose a similar linear program for finding approximately optimal mechanisms. The variables in their LP are the same: the interim allocations and payments. In their framework, however, finding interim allocation rules that can be induced by ex-post feasible allocation rules is crucial. To do so, their constraints, even for simple feasibility sets, are exponential (see \cite{border2007reduced}). To solve their LP efficiently they show how to construct (efficient) separation oracles. Once interim allocation and payment rules are found, they use complicated techniques --- techniques that are not amenable to online arrivals of agents --- to construct a final mechanism. In contrast, our LP, even though it uses exactly the same variables, only needs to ensure feasibility in expectation. For many settings of interest, our LP has a polynomial-size description, and thus can be solved by any LP solver, making our approach more convenient in practice. Furthermore, our techniques are designed to accommodate for online arrivals of agents. 
\vspace{-3mm}
\paragraph{Flipping a coin with probability $p^*_{i,j}(v_i)$.}
In order to implement step~\eqref{eq: flip coin in tCRS} of~\Cref{algo: tCRS to mech} and step~\eqref{eq: flip coin tOCRS} of~\Cref{algo: tOCRS to mech} we can use a Bernoulli factory for division (e.g., the Bernoulli factory of~\cite{morina2021bernoulli}]), which, given a $(c-\epsilon)$-coin and a $p^*_{i,j}(v_i)$-coin, outputs a $\frac{c-\epsilon}{p^*_{i,j}(v_i)}$ using an expected number of coin tosses that is at most $O((p^*_{i,j}(v_i) - c + \epsilon)^{-1})$. We know $c - \epsilon$ exactly, so the $(c-\epsilon)$-coin can be implemented trivially. One can implement a $p^*_{i,j}(v_i)$-coin as follows: For each $i' \neq i$ sample $\hat{r}_{i'} \sim \Dcal_{i'}$ and let $R_{i',j} \gets 1$ with probability $b \, \pi_{i',j}(\hat{r}_{i'})$. Also let $R_{i,j} = 1$ with probability $b \, \pi_{i,j}(v_i)$. Querying the tCRS/tOCRS $\mu$ on the active set $R$ returns a set $Z$ of selected elements; output $Z_{i,j}$ as the coin flip for the $p^*_{i,j}(v_i)$-coin.

\section{Constructing tCRS and tOCRS}\label{sec:constructing tCRS/tOCRS}

In this section, we show how to construct tCRSs and tOCRSs for various feasibility constraints. First, we prove that for a general family of constraints we call  \emph{Vertical-Horizontal (VH) constraints}  (which allow, e.g., constraints like ``every item can be allocated to at most one agent'' or ``the set of items allocated to agent $i$ must form a matroid''), it is possible to construct tCRSs (resp. tOCRSs) given CRSs (resp. OCRSs) in a black-box manner. 

\begin{definition}[VH Constraints]\label{dfn: vh consrtaints}
    We call $\Fcal$ a Vertical-Horizontal (VH) constraint with respect to a ground set $N = \{ e_{i,j} \}_{i \in [n], j \in [m]}$ if there exist sets of constraints $\{\Fcal_i\}_{i \in [n]}$, $\{\Fcal^j\}_{j \in [m]}$ such that $I \in \Fcal$ iff (i) $\forall i \in [n]$, $I\cap \{ e_{i,j} \}_{j \in [m]} \in \Fcal_i$, (ii) $\forall j \in [m]$, $I\cap \{ e_{i,j} \}_{i \in [n]} \in \Fcal^j$.
\end{definition}

\begin{theorem} \label{theorem:vertical-horizontal}
    Given $(b,c)$-selectable CRSs (resp. OCRSs) for constraints $\{\Fcal_i\}_{i \in [n]}$ and $(b,c')$-selectable CRSs (resp. OCRSs) for constraints $\{\Fcal^j\}_{j \in [m]}$, there exists a ($b,c\cdot c'$)-selectable tCRS (resp. tOCRS) for the induced Vertical-Horizontal constraint $\Fcal$.
\end{theorem}

\begin{proof}
    Let $(\Dcal,b \,x)$ be the two-level stochastic process through which elements become active. 
    Let $\mu_{i}$ be the ($b,c$)-selectable CRS/OCRS for constraint $\Fcal_i$, for $i \in [n]$, and let $\mu^{j}$ be the ($b,c'$)-selectable CRS/OCRS for constraint $\Fcal^j$, for $j \in [m]$. Given a set of active elements $R(\Dcal, b\, x)$ sampled from $(\Dcal, b\, x)$, our tCRS/tOCRS selects element $e_{i,j} \in R(\Dcal, b\, x)$ if (i) $\mu^{j}$ on input $R(\Dcal, b\, x)\cap \{ e_{i,j} \}_{i \in [n]}$ selects $e_{i,j}$, and (ii) $\mu_{i}$ on input $R(\Dcal, b\, x)\cap \{ e_{i,j} \}_{j \in [m]}$ selects $e_{i,j}$ (noting that this is process does make decisions online when $\mu^{j}$ and $\mu_i$ are OCRSs and are queried in an online fashion).

    
    Let $A_{i,j}$ be the event that $e_{i,j} \in R(\Dcal,b\, x)$. 
    By the definition of a two-level stochastic process, event $A_{i,j}$ is independent from event $A_{i',j}$, for all $j \in [m]$ and $i, i' \in [n]$ such that $i \neq i'$. Therefore, the CRSs/OCRSs  $\mu^{j}$, $j \in [m]$, are queried about elements that become active independently (as required by the definition of a CRS/OCRS). Now, overloading notation, let $A_{i,j}(d_i)$ be the event that $e_{i,j} \in R(\Dcal,b\, x)$ given that $d_i$ was sampled from $\Dcal_i$. By the definition of a two-level stochastic process, event $A_{i,j}(d_i)$ is independent from event $A_{i,j'}(d_i)$, for all $j \neq j' \in [m]$ and all $i \in [n]$. Therefore, the CRSs/OCRSs  $\mu_{i}$, $i \in [n]$, are queried about elements that become active independently (as required by the definition of a CRS/OCRS). Let $B_{i,j}$ be the event that $\mu^{j}$ selects an active element $e_{i,j} \in R(\Dcal, b\, x)$ on input $R(\Dcal, b\, x)\cap \{ e_{i,j} \}_{i \in [n]}$, and note that, since $\mu^{j}$ is a ($b,c'$)-selectable CRS/OCRS, we have that then $\Pr[B_{i,j}|A_{i,j}] \ge c'$. Similarly, let $C_{i,j}$ be the event that $\mu_{i}$ selects an active element $e_{i,j} \in R(\Dcal, b\, x)$ on input $R(\Dcal, b\, x)\cap \{ e_{i,j} \}_{j \in [m]}$; $\Pr[C_{i,j}|A_{i,j}] \ge c$. Finally, conditioned on $A_{i,j}$ events $B_{i,j}$ and $C_{i,j}$ are conditionally independent due to the definition of a two-level stochastic process. Thus, $\Pr[B_{i,j} \cap C_{i,j}|A_{i,j}] \ge c \, c'$, which concludes the proof. 
\end{proof}


By combining~\Cref{theorem:vertical-horizontal} with known results (e.g.,~\cite{hajiaghayi2007automated,lee2018optimal,kashaev2023simple}) we can get tCRSs and tOCRSs for various settings of interest; we show these applications in~\Cref{thm: from t to mechanisms} in~\Cref{sec:applications}. For the remainder of this section, we construct tOCRSs for specific knapsack constraints (that are not VH constraints, and therefore~\Cref{thm: from t to mechanisms} is not applicable). We start by defining knapsack and multi-choice knapsack.

\begin{definition}[Knapsack Constraints]
   Consider a ground set of elements $N = \{ e_{i,j} \}_{i \in [n], j \in [m]}$, where, for each $i , j \in [n] \times [m]$, element $e_{i,j}$ has a weight $k_{i,j}$, and there is a maximum weight $K$. We say that $\Fcal$ is a Knapsack constraint when $I \in \Fcal$ if and only if $I \subseteq N$, and $\sum_{(i,j): e_{i,j} \in I} k_{i,j} \le K$.
\end{definition}

\begin{definition}[Multi-Choice Knapsack]
   Consider a ground set of elements $N = \{ e_{i,j} \}_{i \in [n], j \in [m]}$, where, for each $i , j \in [n] \times [m]$, element $e_{i,j}$ has a weight $k_{i,j}$, and there is a maximum weight $K$. We say that $\Fcal$ is a Multi-Choice Knapsack constraint when $I \in \Fcal$ if and only if $I \subseteq N$, $\sum_{(i,j): e_{i,j} \in I} k_{i,j} \le K$, and, for all $i \in [n]$, $|I\cap \{e_{i,j}\}_{j \in [m]}| \le 1$.
\end{definition}

The following theorem states gives a $(b, \frac{1}{2+8b})$-selectable tOCRS for Knapsack Constraints, for $b \in [0,1]$. Interestingly, since tOCRSs are OCRSs, our result implies a $(1,0.1)$-selectable OCRS for knapsack; this is better than the $(1,0.085)$-selectable OCRS given by Feldman et al.~\cite{feldman2021online}, but not as good as the state-of-the-art $(1,1/(3+e^{-2}))$-selectable ($\approx (1,0.319)$-selectable) OCRS of Jiang et al.~\cite{jiang2022tight}.

\begin{theorem}
\label{theorem:knapsack}
    For all $b \in [0,1]$, there exists a $(b, \frac{1}{2+8b})$-selectable tOCRS for Knapsack.
\end{theorem}

\begin{proof}
 Let $(\Dcal,b x)$ be the two-level stochastic process through which elements become active. Let $k_{i,j}$ be the weight of element $e_{i,j}$, and $K$ bet the total weight. Let $S_h = \{e_{i,j}: i\in [n], j \in [m]| k_{i,j}>K/2\}$ be the set of elements whose weight is at least $K/2$, the ``heavy elements,''  and let
 $S_{\ell} = \{e_{i,j}: i\in [n], j \in [m]| k_{i,j} \le K/2\}$ be the set of ``light elements.''  Our tOCRS is randomized: with probability $1/2$ we run a scheme that considers only heavy elements, the ``heavy scheme,'' and with probability $1/2$ we run a scheme that considers only light elements, the ``light scheme.'' In both cases, we use $I$ to indicate the set of elements we output. Without loss of generality (from the definition of a tOCRS) we assume that elements arrive in the order $e_{1,1}, e_{1,2}, \dots, e_{n,m}$.

 For the case of the heavy scheme, it is obvious that we can only take one heavy element. We initialize $I= \emptyset$ and consider elements sequentially. Let $A_{i,j}(d_i)$ be the event that $I$ is empty until element $e_{i,j}$ is considered when we run the heavy scheme and $d_i$ was sampled from the two-level stochastic process $(\Dcal,b x)$. If element $e_{i,j}$ is active, $e_{i,j} \in S_h$ and $I= \emptyset$, then with probability $\frac{1}{(1+4b)\Pr[A_{i,j}(d_i)]}$ we set $I \leftarrow \{e_{i,j}\}$, otherwise we move on to the next element. Assuming that $\frac{1}{(1+4b)\Pr[A_{i,j}(d_i)]}$ is a valid probability, each heavy element is selected with probability exactly $\Pr[A_{i,j}(d_i)]\frac{1}{(1+4b)\Pr[A_{i,j}(d_i)]} = \frac{1}{(1+4b)}$ when we run the heavy scheme. 
 Towards proving that $\Pr[A_{i,j}(d_i)] \ge \frac{1}{(1+4b)}$, assume, let $\ell_i = \mathrm{argmin} \{ j \in [m] | e_{i,j} \in S_h \}$ and $u_i = \mathrm{argmax} \{ j \in [m] | e_{i,j} \in S_h \}$ are the lowest and highest index heavy elements for agent $i$. For $e_{i,j} \in S_h$, we have that 

\begin{align*}
    \Pr[A_{i,j+1}(d_i)] &= \Pr[A_{i,j}(d_i) \cap e_{i,j} \notin I] \\
    &= \Pr[A_{i,j}(d_i)] \Pr[e_{i,j} \notin I | A_{i,j}(d_i)] \\
    &= \Pr[A_{i,j}(d_i)] \left( \underbrace{1 - bx_{i,j}(d_i) \vphantom{\left(1 - \frac{1}{(1+4b)\Pr[A_{i,j}(d_i)]}\right)}}_{\text{$e_{i,j}$ not active} } + \underbrace{bx_{i,j}(d_i)\left(1 - \frac{1}{(1+4b)\Pr[A_{i,j}(d_i)]}\right)}_{\text{$e_{i,j}$ active but not selected}} \right) \\
    &= \Pr[A_{i,j}(d_i)] - \frac{b}{1+4b} x_{i,j}(d_i).
\end{align*}

Recursing, we have $\Pr[A_{i,j+1}(d_i)] = \Pr[A_{i,\ell_i}(d_i)] - \frac{b}{1+4b} \sum_{\substack{j' \leq j: e_{i,j'} \in S_h}} x_{i,j'}(d_i)$. Towards computing $\Pr[A_{i,\ell_i}(d_i)]$, for $i \geq 1$ we have:

\begin{align*}
    &\Pr[A_{i+1,\ell_{i+1}}(d_{i+1})] = \sum_{d_i \in \Vcal_{i}} \Pr[ \Dcal_i = d_i] \Pr[A_{i,u_i}(d_i) \cap e_{i,u_i} \notin I | \Dcal_i = d_i]\\
    &= \sum_{d_i \in \Vcal_{i}} \Pr[\Dcal_i = d_i]\left( \Pr[A_{i,\ell_i}(d_i)] - \frac{b}{1+4b} \sum_{\substack{j': e_{i,j'}\in S_h}} x_{i,j'}(d_i) \right)\\
    &= \Pr[A_{i,\ell_i}(d_i)] - \frac{b}{1+4b} \sum_{\substack{j': e_{i,j'}\in S_h}} \sum_{d_i \in \Vcal_{i}} \Pr[\Dcal_i = d_i] \, x_{i,j'}(d_i) \tag{$A_{i,\ell_i}(d_i)$ does not depend on $d_i$}\\
    &= 1- \frac{b}{1+4b} \sum_{i' \leq i} \sum_{\substack{j': e_{i',j'}\in S_h}} w_{i,j}, \tag{telescoping}   
\end{align*}

  where $w_{i,j} = \sum_{d_{i} \in \Vcal_i} \Pr[\Dcal_{i} = d_{i}] \, x_{i,j}(d_{i})$, is the probability that $e_{i,j}$ is active in $(\Dcal,b x)$. Since $K \geq \sum_{\substack{i, j: e_{i,j} \in S_h}} x_{i,j}(d_i) k_{i,j} > \sum_{\substack{i, j: e_{i,j} \in S_h}} x_{i,j}(d_i) \cdot \frac{K}{2}$, we have $\sum_{\substack{i, j: e_{i,j} \in S_h}} x_{i,j}(d_i) < 2$. Hence, it also holds that 
$\sum_{\substack{i,j: e_{i,j} \in S_h}} w_{i,j} < 2$. Combining with the above we have that:
\vspace{-3mm}
\begin{align*}
    \Pr[A_{i,j+1}(d_i)] &= 1- \frac{b}{1+4b} \left(\sum_{i' \leq i} \sum_{\substack{j': e_{i,j'}\in S_h}} w_{i',j'} +  \sum_{\substack{j' < j+1:\\e_{i,j'}\in S_h}} x_{i,j'}(d_i)\right) > 1-\frac{4b}{1+4b} = \frac{1}{1+4b}.
\end{align*}

For the case of the light scheme, we again initialize $I = \emptyset$. Let $B_{i,j}(d_i)$ be the event that, at the time step when element $e_{i,j}$ is considered, the total weight of elements in $I$ so far is strictly less than $K/2$, when $d_i$ was sampled from the two-level stochastic process $(\Dcal, bx)$. We consider each (light) element $e_{i,j}$ one at a time, and if $e_{i,j}$ is active and the weight of elements in $I$ is less than $K/2$, we set $I \leftarrow I \cup \{e_{i,j}\}$ with probability $\frac{1}{(1+4b)\Pr[B_{i,j}(d_i)]}$; otherwise we move on to the next element. Again, if $\frac{1}{(1+4b)\Pr[B_{i,j}(d_i)]}$ is a valid probability, each light element is selected with probability exactly $\Pr[B_{i,j}(d_i)]\frac{1}{(1+4b)\Pr[B_{i,j}(d_i)]} = \frac{1}{(1+4b)}$ when we run the light scheme. It therefore remains to prove that $\Pr[B_{i,j}(d_i)] \ge \frac{1}{(1+4b)}$. Let $W_{i,j}(d_i)$ be the random variable that represents the total weight of elements in $I$ at the time when we consider elements $e_{i,j}$ when $d_i$ was the sample from $(\Dcal, bx)$. We have that:
 \begin{align*}
     \Ex{W_{i,j}(d_i)} &= \underbrace{\sum_{i'<i} \sum_{d_{i'} \in \Vcal_{i'}} \Pr[\Dcal_{i'} = d_{i'}] \sum_{j': e_{i',j'} \in S_{\ell}} \frac{1}{(1+4b)\Pr[B_{i',j'}(d_{i'})]} \Pr[B_{i',j'}(d_{i'})]bx_{i',j'}(d_{i'})k_{i',j'}}_{\text{Contribution from agents before $i$}} \\
     &+\underbrace{\sum_{\substack{j'<j:\\ e_{i,j'} \in S_{\ell}}} \frac{1}{(1+4b)\Pr[B_{i,j'}(d_{i})]}\Pr[B_{i,j'}(d_{i})] bx_{i,j'}(d_{i})k_{i,j'}}_{\text{Contribution from $i$'s items before $j$}}\\
     &= \frac{b}{1+4b}\left(\sum_{i'<i} \sum_{j': e_{i',j'} \in S_{l}} w_{i',j'}k_{i',j'}+  \sum_{\substack{j'<j:\\ e_{i,j'} \in S_{l}}} x_{i,j'}(d_{i})k_{i,j'}\right) \\
     &\le  \frac{2b}{1+4b}K. \tag{Feasibility of $(\Dcal,bx)$}
\end{align*}
Therefore, we have
\begin{align*}
    \Pr[B_{i,j}(d_i)] &= \Pr[W_{i,j}(d_i) < K/2] \\
    &= 1-\Pr[W_{i,j}(d_i) \ge K/2] \\
    &\ge 1 - \frac{\frac{2b}{1+4b}K}{K/2} \tag{Markov's Inequality}\\
    &= \frac{1}{1+4b}.
\end{align*}
Since we run each scheme with probability $1/2$, for an element $e_{i,j} \in S_{h}$ that is active we have:
\begin{align*}
    \Pr[e_{i,j} \in I] &= \Pr[\text{``heavy scheme''}] \Pr[e_{i,j} \in I|\text{``heavy scheme''}] \ge \frac{1}{2}\frac{1}{1+4b} = \frac{1}{2+8b}.
\end{align*}
Similarly, for an active element $e_{i,j} \in S_{\ell}$, $\Pr[e_{i,j} \in I] \geq 1/(2+8b)$, concluding the proof.
\end{proof}

Next, we give a $(b, \frac{1}{2+7b})$-selectable tOCRS for Multi-Choice constraints, for all $b \in [0,1]$. The proof of the following theorem is deferred to~\Cref{app: missing from tCRS const}.

\begin{theorem} \label{theorem:knapsack unit demand}
    For all $b \in [0,1]$, there exists a $(b, \frac{1}{2+7b})$-selectable tOCRS for Multi-Choice Knapsack.
\end{theorem}

\subsection{Efficient Implementation Considerations}

~\Cref{theorem:knapsack} and~\Cref{theorem:knapsack unit demand} show that ``Knapsack'' and ``Multi-Choice Knapsack'' tOCRSs exist. Both tOCRSs have non-constructive coin-flipping steps (e.g., selecting an active light element with probability $\frac{1}{(1+4b)\Pr[B_{i,j}(d_i)]}$, where $B_{i,j}(d_i)$ is the event that, at the time step when element $e_{i,j}$ is considered, the total weight of elements in $I$ so far is strictly less than $K/2$, when $d_i$ was sampled from the two-level stochastic process $(\Dcal, bx)$. In this section, we show how to efficiently implement these steps, albeit with a small loss in the performance guarantee.

\begin{proposition} \label{proposition:efficient implimentation knapasack}
    We can implement a $\left( b, \frac{1}{2+8b}\left(\frac{1-\delta}{1+10 \epsilon}\right) \right)$-selectable tOCRS for Knapsack in time $poly(1/\epsilon^2, \log (1/\delta), m, n)$.
\end{proposition}

\begin{proof}
    In the proof of~\Cref{theorem:knapsack} we give an exact formula for the probability that $A_{i,j}(d_i)$ occurs; therefore, the only step we cannot directly implement from the procedure outlined in the proof of~\Cref{theorem:knapsack} is the toss of the $\frac{1}{(1+4b)\Pr[B_{i,j}(d_i)]}$ coin. Even when given a $\Pr[B_{i,j}(d_i)]$-coin, using a Bernoulli factory for division to produce a $\frac{1}{(1+4b)\Pr[B_{i,j}(d_i)]}$-coin results in an exponential blow-up in computation, since the factory for the $k$-th coin would need to also simulate the factory for the $(k-1)$-st coin, and so on. Instead, we approximate these probabilities, sequentially, using multiple experiments and bounding the error using Chernoff bounds. 

    In order to decide whether to select some element $e_{i,j} \in S_{\ell}$, we repeatedly simulate our algorithm until element $e_{i,j}$, for $T = \frac{1}{2\epsilon^2}\log \frac{2nm}{\delta}$ repetitions, where the choice of running the ``light scheme'' and $d_i$ are fixed. In this simulation, the coins needed to make decisions until element $e_{i,j}$ are replaced with estimated coins (described shortly). Let $X_t$ be the indicator random variable for the event that $B_{i,j}(d_i)$ occurred at simulation $t \in [T]$. We select element $e_{i,j}$ (when it is active) with probability $\frac{1}{(1+4b)\left(\frac{1}{T}\sum_{t \in [T]} X_t+ \epsilon\right)}$.
    Standard Chernoff–Hoeffding bounds~\cite{hoeffding1994probability} imply that
    \begin{align*}
        \Pr\left[ \left|\frac{1}{T}\sum_{t \in [T]} X_t - \Pr[B_{i,j}(d_i)] \right|>\epsilon \right]& \le 2\exp\left(-2\epsilon^2 T \right) =2\exp\left(-2\epsilon^2 \frac{1}{2\epsilon^2}\log \frac{2nm}{\delta} \right) \le \frac{\delta}{nm}.
    \end{align*}
    Assuming that $\left|\frac{1}{T}\sum_{t \in [T]} X_t - \Pr[B_{i,j}(d_i)] \right| \le \epsilon$, $\left(\frac{1}{T}\sum_{t \in [T]} X_t+ \epsilon \right) \in [\Pr[B_{i,j}(d_i)], \Pr[B_{i,j}(d_i)]+2\epsilon]$. 
    
    When bounding $\Ex{W_{i,j}(d_i)}$ (in the proof of~\Cref{theorem:knapsack}), the term $\frac{1}{(1+4b)\Pr[B_{i,j}(d_i)]} \cdot \Pr[B_{i,j}(d_i)]$ is replaced with $\frac{\Pr[B_{i,j}(d_i)]}{(1+4b)\left(\frac{1}{T}\sum_{t \in [T]} X_t+ \epsilon\right)}$, which is at most $\frac{1}{1+4b}$. Thus $\Ex{W_{i,j}(d_i)} \le \frac{2b}{1+4b}K$, which gives us $\Pr[B_{i,j}(d_i)]\ge \frac{1}{1+4b} \ge 1/5$ using the same arguments presented in the proof of \Cref{theorem:knapsack}. Thus, the probability that an active, light element $e_{i,j}$ is selected is
    \[\frac{\Pr[B_{i,j}(d_i)]}{(1+4b)\left(\frac{1}{T}\sum_{t \in [T]} X_t+ \epsilon\right)} \ge \frac{1}{1+4b} \left(\frac{\Pr[B_{i,j}(d_i)]}{\Pr[B_{i,j}(d_i)]+2\epsilon} \right) \ge \frac{1}{1+4b} \left(\frac{1}{1+10 \epsilon} \right).\]
    Using a union bound we have that
    \[\Pr\left[ \exists (i,j) \in [n]\times [m]: \left|\frac{1}{T}\sum_{t \in [T]} X_t - \Pr[B_{i,j}(d_i)] \right|>\epsilon \right] \le \delta.\]
    Thus, we overall have that with probability at least $1-\delta$, when we run the light scheme, each active element will be selected with probability at least $\frac{1}{1+4b} \left(\frac{1}{1+10 \epsilon} \right)$.   
\end{proof}

\begin{proposition}\label{prop: efficient knapsack unit demand}
    We can implement a $\left( b, \frac{1}{2+7b}\left(\frac{1-\delta}{1+8 \epsilon}\right) \right)$-selectable tOCRS for Multi-Choice Knapsack in time $poly(1/\epsilon^2, \log (1/\delta), m, n)$.
\end{proposition}

The proof of~\Cref{prop: efficient knapsack unit demand} is similar to the proof of~\Cref{proposition:efficient implimentation knapasack} and is deferred to~\Cref{app: missing from tCRS const}.
\section{Applications}\label{sec:applications}

In this section, we show how to combine results from Sections~\ref{sec: mech from tCRS} and~\ref{sec:constructing tCRS/tOCRS} to get (sequential) mechanisms for various problems of interest.

First, our framework can give a sequential mechanism with a $\frac{2e}{e-1}$ ($\approx 3.16$) approximation guarantee for the problem of auctioning off $m$ items to $n$ agents with additive preferences, under ``matroid Vertical-Horizontal constraint'' constraints $\Fcal$: (i) the set of agents that each item $j$ is allocated to must form a matroid, and (ii) the set of items allocated to each agent $i$ must form a matroid. Observe that, $\Fcal$ is a VH constraint (\Cref{dfn: vh consrtaints}), induced by the aforementioned constraints (i) and (ii). Both (i) and (ii) are matroid constraints,~\cite{chekuri2014submodular} give a $(1,1-\frac{1}{e})$-selectable CRS and~\cite{lee2018optimal} give a $(1,1/2)$-selectable OCRS for matroid constraints. Therefore,~\Cref{theorem:vertical-horizontal} (where we use the CRS for the constraints over items and the OCRS for the constraints over agents) implies a $(1,\frac{e-1}{2e})$-selectable tOCRS for $\Fcal$:

\begin{corollary}[\Cref{theorem:vertical-horizontal} and~\cite{lee2018optimal}] \label{corollary: matroid}
    There exists a $(1,\frac{e-1}{2e})$-selectable tOCRS for matroid Vertical-Horizontal constraint constraints.
\end{corollary}

\Cref{corollary: matroid} and~\Cref{thm: from t to mechanisms} readily give the following result.

\begin{application}[\Cref{corollary: matroid} and \Cref{thm: from t to mechanisms}]\label{application: matroid constraint}
    Consider the problem of auctioning off $m$ items to $n$ agents with additive preferences, such that the set of agents that each item $j$ is allocated to must form a matroid, and the items allocated to each agent must form a matroid. There exists a sequential, BIC-IR mechanism that guarantees a $\frac{2e}{e-1}$ ($\approx 3.16$) approximation to the optimal revenue.
\end{application}

Notably, the previously best-known approximation guarantee for a sequential mechanism for even a special case of this problem (agents with preferences that are constrained additive with a matroid constraint, and each item can be allocated to at most one agent) was $70$~\cite{cai2017subadditive}, which we improve to $\frac{2e}{e-1} \approx 3.16$.

CRSs and OCRSs with approximation factors better than $2$ (i.e. better than the result of~\cite{lee2018optimal} for general matroids) are possible for some special cases. 
E.g., for $k$-uniform matroids,~\cite{hajiaghayi2007automated} give a $\left( 1 - O\left( \sqrt{\frac{\log k}{k}} \right) \right)$-selectable OCRS, and~\cite{kashaev2023simple} 
give a $\left( 1, 1 - \binom{n}{k} \left( 1- \frac{k}{n} 
\right)^{n+1-k} \left(\frac{k}{n}\right)^k \right)$-selectable CRS, 
where $n$ is the number of elements (for a fixed $k$, this approaches $\left(1 , 1 - e^{-k} \frac{k^k}{k!}\right)$). Combining with~\Cref{theorem:vertical-horizontal} we can get tOCRSs/tCRSs for these cases, and applying~\Cref{thm: from t to mechanisms} gives an overall (sequential for tOCRSs) mechanism with a slightly improved guarantee.

We note that, depending on the choice of matroids,~\eqref{lp} might not be efficiently computable. For example, the representation of $P_{\Fcal}$ might require an exponential (in $n$, $m$ and $\sum_{i \in [n]} |\Vcal_i|$) number of inequalities; in such cases, our approach does not give an end-to-end efficient procedure for finding a mechanism. However, if one is given feasible in expectation, BIC-IR and approximately optimal interim rules, all remaining steps of our framework can be efficiently computed.



As a second application, we combine~\Cref{thm: from t to mechanisms} with our tOCRSs for Knapsack and Multi-Choice Knapsack.

\begin{application}[\Cref{theorem:knapsack},~\Cref{theorem:knapsack unit demand}, and \Cref{thm: from t to mechanisms}] \label{application: Knapsack}
    Consider the problem of auctioning off $m$ items to $n$ agents with additive preferences, where each item $j \in [m]$ has some weight $k_j$ and the total weight of items sold cannot exceed $K$. 
    There exists an efficiently computable, sequential, BIC-IR mechanism that guarantees a $10$ approximation to the optimal revenue. Additionally, if each agent $i$ can get at most one item, there exists an efficiently computable, sequential, BIC-IR mechanism that guarantees a $9$ approximation to the optimal revenue.
\end{application}

Recently, \cite{aggarwal2022simple} gave a simple and approximately optimal, with respect to welfare, mechanism for the Rich-Ads problem. \Cref{application: Knapsack} readily gives a sequential, approximately optimal, with respect to revenue, mechanism for the same problem.

Finally, by creating a meta-item for each possible subset of items (where the weight of the meta-item is simply the sum of weights from the corresponding subset), our approach gives an approximately optimal, sequential, BIC-IR mechanism for \emph{arbitrary} valuation functions (subject to a knapsack constraint). For a logarithmic (in $n$) number of items, this approach gives a computationally efficient procedure as well (but, of course, this is not true in general).

\begin{application}[\Cref{theorem:knapsack unit demand} and \Cref{thm: from t to mechanisms}]\label{application: arbitrary}
    Consider the problem of auctioning off $m$ items to $n$ agents with arbitrary valuation functions, where each item $j \in [m]$ has some weight $k_j$ and the total weight of items sold cannot exceed $K$. There exists a sequential, BIC-IR mechanism that guarantees a $9$ approximation to the optimal revenue.
\end{application}

\section{Extensions to procurement auctions}\label{sec: procurement}

In this section, we show how to extend our framework for the case of procurement auctions.
We only show how our framework works for sequential procurement auctions, where we construct an auction using an OCRS (and interim allocations/payments); our results can be extended to give non-sequential auctions using a CRS, similarly to our results in~\Cref{sec: mech from tCRS}.

Budget feasible procurement auctions were introduced by the seminal work of Singer~\cite{singer2010}. Following this, there has been a line of work studying deterministic and randomized budget feasible mechanisms that obtain approximately optimal welfare, where a major focus has been on single dimensional and prior-free settings~\cite{ChenGL11,GravinJLZ20,AnariGN18,klumper2022budget}.

We start by defining the procurement problem we study.

\paragraph{Procurement Preliminaries}

There is a single buyer and a set of $n$ sellers. Each seller has a total of $m$ services they can provide. The buyer has a publicly known value $v_{i,j}$ for getting the $j$-th service that seller $i$ offers. An (integral) allocation $x \in \{ 0 , 1 \}^{nm}$ indicates which services the buyer received: $x_{i,j} \in \{ 0 ,1 \}$ is the indicator for whether the buyer received the $j$-th service of seller $i$. The buyer's value for an allocation $x$ is $\sum_{i \in [n], j\in [m]} x_{i,j} v_{i,j}$. The buyer will pay the sellers for their services; the buyer's objective is to maximize her value without paying more than a (publicly known) budget $B$. Each seller $i \in [n]$ has some cost for providing service $j \in [m]$ depicted as $c_{i,j}$. We will assume that $c_i$ is drawn from a known distribution $\Ccal_i$; we allow for correlation between the cost for different services for a fixed seller, but require independence between sellers.

A procurement auction elicits reported costs $(c_1, \dots, c_n)$, and determines which services are procured from which seller, as well as the payments to the sellers. Our goal is to design BIC-IR, budget-feasible procurement auctions that maximize the buyer's expected value. The definition of BIC-IR, approximate optimality, sequentiality, and interim allocations/payments are similar to the corresponding definitions from~\Cref{sec:prelims}, and are deferred to~\Cref{app: missing procurement prelims}.

For ease of exposition, we will assume that there are no constraints on the services we can acquire, other than the buyer's budget constraint. If additional constraints exist, our framework can be extended using the ideas analyzed in the previous sections.

\subsection{Procurement Framework}

Our procurement framework uses OCRSs for Stochastic Knapsack. A $c$-selectable OCRS for Stochastic Knapsack $\mu_K$ is parameterized by a knapsack size $K$ and distributions from which elements' weight are drawn. The OCRS is given, in an online manner, elements and their weight (which is drawn from the aforementioned distributions), one at a time, and it needs to decide, immediately and irrevocably, whether to include an element in the knapsack, in a way that every element is selected with (ex-ante) probability at least $c$; see~\Cref{sec: stochastic knapsack ocrs} for more details.

The input to our framework is  (i)  a feasible in expectation, BIC-IR interim allocation, payment rule pair $(\pi,q)$ that is an $\alpha \geq 1$ approximation, (ii) a $c$-selectable OCRS for Stochastic Knapsack, and (iii) a parameter $\epsilon \geq 0$.  Our framework,~\Cref{algo: OCRS to proc}, outputs a BIC-IR procurement auction that is a $\alpha /  (c - \epsilon)$ approximately optimal. 

\Cref{algo: OCRS to proc} works as follows. We approach each seller $i$ sequentially. Seller $i$ reports $r^*_i \in supp(\Ccal_i)$, and we query the OCRS on input $q_i(r^*_i)$. If the OCRS selects an element with weight $q_i(r^*_i)$, we pay seller $i$ an amount equal to $q_i(r^*_i)$, with a certain probability (this step ensures that the expected payment to seller $i$ is exactly $(c-\epsilon) q_i(r^*_i)$). Finally, for each service $j \in [m]$, we receive the service from seller $i$ with probability $(c-\epsilon) \pi_{i,j}(r^*_i)$.

\begin{algorithm}[!ht]
  \DontPrintSemicolon
  \KwInput{an interim allocation, payment rule pair $(\pi,q)$, a $c$-selectable OCRS $\mu_B$ for Stochastic Knapsack, a parameter $\epsilon \geq 0$.}
   \hrulealg
  \For{each seller $i \in [n]$}{
        Seller $i$ reports $r^*_i \in supp(\mathcal{C}_i)$.\\
        $Z_i \gets \mu_B(q_i(r^*_i))$.\\
        \If{$Z_{i} = 1$}{
            Pay seller $i$,  $q_i(r^*_i)$ with probability $(c - \epsilon)/p^*_{i}(r^*_i)$, where $p^*_{i}(r^*_i)$ be the probability that the OCRS selects an element with weight $q_i(r^*_i)$.
        }
        \For{each service $j \in [m]$}{
            Receive service $j$ from seller $i$ with probability $(c-\epsilon)\pi_{i,j}(r^*_i)$
        }  
    }
\caption{Our sequential procurement auction when given an OCRS}\label{algo: OCRS to proc}
\end{algorithm}

\begin{theorem}\label{thm: from OCRS to proc}
    Given (i) feasible in expectation, BIC-IR interim allocation and payment rules $(\pi,q)$ that are an $\alpha \geq 1$ approximation, (ii) a $c$-selectable OCRS for Stochastic Knapsack, and (iii) a parameter $\epsilon \geq 0$,~\Cref{algo: OCRS to proc} gives a BIC-IR sequential procurement auction that is $\alpha /(c - \epsilon)$-approximately optimal. If we have query access to the OCRS, our mechanism can be implemented using a $O(poly( \sum_i |supp(\Ccal_i)|, m, \frac{1}{\epsilon}))$ number of queries in expectation.
\end{theorem}

\begin{proof}
First, we argue that~\Cref{algo: OCRS to proc} is budget feasible with probability $1$. By definition, and assuming truthful reports, the interim payment $q$ defines a feasible distribution of ``weights'' for each seller. The OCRS always selects a set of elements whose weight is at most the knapsack size (in our case $B$), and our total payments are at most the total weight that the OCRS packs in the knapsack; therefore, our total payments are at most $B$.

Second, we argue that~\Cref{algo: OCRS to proc} is BIC-IR. From the perspective of seller $i$, a report $r_i \in \Vcal_i$ costs  $\frac{c-\epsilon}{p^*_{i}(r_i)} p^*_{i}(r_i) q_i(r_i) = (c-\epsilon)q_i(r_i)$. The expected cost of services is $\sum_{j \in [m]} c_{i,j} (c-\epsilon) \pi_{i,j}(r_i)$. Therefore, her expected utility is $(c - \epsilon)\left(q_i(r_i) -\sum_{j \in [m]} c_{i,j} \pi_{i,j}(r_i)\right)$; since $(\pi, q)$ is BIC, so is~\Cref{algo: OCRS to proc}. Near-identical arguments imply the BIR guarantee. 

The expected value of the buyer is $\sum_{i \in [n]} \sum_{c_i \in \Ccal_i} \Pr[c_i] \sum_{j \in [m]} v_{i,j} \, (c-\epsilon) \pi_{i,j}(c_i)$, which is a $\alpha / (c-\epsilon)$ approximation, since $(\pi,q)$ is an $\alpha$ approximation.

If we are given only black-box access to an OCRS for Stochastic Knapsack, it is not immediately straightforward how to flip a coin with probability $(c - \epsilon)/p^*_{i}(r^*_i)$ (efficiently or otherwise), as needed in ~\Cref{algo: OCRS to proc}. Using a Bernoulli factory for division (such as the result of~\cite{morina2021bernoulli} discussed in~\Cref{sec:prelims}), we can implement this step with $O(\frac{1}{\epsilon})$ calls in expectation; we discuss implementation details in~\Cref{subsec: implementation proc}.
\end{proof}

\cite{jiang2022tight} give a $\frac{1}{3+e^{-2}}$-selectable OCRS for Stochastic Knapsack. Combining with~\Cref{thm: from OCRS to proc} we readily get the following application.

\begin{application}[\Cref{thm: from OCRS to proc} and~\cite{jiang2022tight}]\label{application: procurement}
    Consider the problem of purchasing $m$ services from $n$ strategic sellers, subject to a budget constraint. There exists a sequential, budget-feasible, BIC-IR procurement auction that guarantees a $3+e^{-2}$ ($\approx 3.13$) approximation to the expected value of the optimal BIC-IR auction.
\end{application}

In~\Cref{sec: stochastic knapsack ocrs} we give a new OCRS for Stochastic Knapsack that is $\max\{\frac{1-k^*}{2-k^*}, \frac{1}{6} \}$-selectable, where $k^* = \frac{1}{K} \max_{i \in [n] , k_i \in supp(\Kcal_i)} \, k_i$, and $\Kcal_i$ is the distribution of weights for the $i$-th element. Our OCRS outperforms the OCRS of~\cite{jiang2022tight} when $k^*$ is small (specifically, $k^* \leq 1/3$). Furthermore, our OCRS induces a greedy and monotone OCRS for the non-stochastic knapsack problem, which is not true for the OCRS of~\cite{jiang2022tight}. Note that, our OCRS also implies that better approximation guarantees are possible for sequential procurement auctions if the payment to a seller is never more than a third of the total budget.

\subsection{Implementation Considerations} \label{subsec: implementation proc}

Here, we highlight some implementation details. First, we give a simple LP that computes optimal ($\alpha = 1$), feasible in expectation, BIC-IR interim allocation and payment rule $(\pi,q)$. Second, we flesh out implementation details regarding flipping a $(c - \epsilon)/p^*_{i}(c_i)$-coin, when given only black-box access to an OCRS.

\paragraph{Finding feasible in expectation, optimal interim rules}
Consider the following linear program,~\eqref{lp:proc}, which computes an interim relaxation of the revenue optimal BIC-IR mechanism.

\begin{equation*}
\label{lp:proc}
\begin{array}{ll@{}ll}
\text{maximize} & \multicolumn{2}{l}{ \displaystyle\sum\limits_{i \in [n]} \displaystyle\sum\limits_{c_i \in \Ccal_i} \Pr[c_i] \sum_{j \in [m]} \pi_{i,j}(c_i) v_{i,j}}\\ \\

\text{s.t.} &  q_i(c_i) - \sum_{j \in [m]} c_{i,j} \pi_{i,j}(c_{i}) & \ge q_i(c'_i) - \sum_{j \in [m]} c_{i,j} \pi_{i,j}(c'_{i}) & \forall i \in [n], c_i, c'_i \in supp(\Ccal_i)\\  
\\
&  q_i(c_i) - \sum_{j \in [m]} c_{i,j} \pi_{i,j}(c_{i}) & \ge 0 & \forall i \in [n], c_i \in supp(\Ccal_i)  \\
\\
& \sum_{i \in [n]} \sum_{c_i \in \Ccal_i} \Pr[c_i] q_i(c_i) & \le B \\
\\
& q_i(c_i) & \le B & \forall i \in [n], c_i \in supp(\Ccal_i)
\end{array}
\hfill \text{(LP2)}
\end{equation*}

This LP has $O(n \sum_{i \in [n]} |supp(\Ccal_i)|)$ variables, and $O(n \sum_{i \in [n]} |supp(\Ccal_i)|^2)$ constraints, and is therefore efficiently computable by standard LP solvers.

\paragraph{Flipping a coin.}
We again use a Bernoulli factory for division to produce a $(c - \epsilon)/p^*_{i}(c_i)$-coin. $c-\epsilon$ is known. And, identically to our approach in~\Cref{subsec: implementation}, we can flip a $p^*_{i}(c_i)$-coin by simulating the entire procedure, conditioning on $c_i$ being the report of seller $i$.

\subsection{A new OCRS for Stochastic Knapsack}\label{sec: stochastic knapsack ocrs}

In the Stochastic Knapsack problem, there is a ground set of elements $N=\{e_i\}_{i \in [n]}$ and a knapsack size $K$. Each element arrives sequentially and reveals a random weight $k_{i} \in [0,K]$ drawn from a known prior distribution $\Kcal_i$ (where a draw of $k_i = 0$ is analogous to element $e_i$ being inactive/not arriving). The input distribution satisfies $\sum_{i \in [n]} \sum_{k_i \in supp(\Kcal_i)}\Pr[\Kcal_i = k_i] \cdot k_i \le K$. Once an element arrives and reveals its weight we need to immediately and irrevocably decide whether this element is included in the knapsack. A $c$-selectable OCRS for this problem is a procedure that selects elements (online), such that the knapsack constraint is never violated (i.e., $\sum_{i \in [n]} k_i \le K$ in all outcomes), and every element is selected with probability at least $c$.

\begin{theorem} \label{theorem:stochastic knapsack}
    There exists a $\max\{\frac{1-k^*}{2-k^*}, \frac{1}{6} \}$-selectable OCRS for Stochastic Knapsack, where $k^* = \frac{1}{K} \max_{i \in [n] , k_i \in supp(\Kcal_i)} \, k_i$.
\end{theorem}

\begin{proof}
    We present two OCRSs: a $\gamma$-selectable OCRS, where $\gamma  = \frac{1-k^*}{2-k^*}$, and a $\frac{1}{6}$-selectable OCRS.
    Our overall OCRS computes $k^* = \frac{1}{K} \max_{i \in [n] , k_i \in supp(\Kcal_i)} \, k_i$. If $\gamma \geq \frac{1}{6}$, it executes the following $\gamma$-selectable OCRS; otherwise, it executes a $\frac{1}{6}$-selectable OCRS.

    First, we give the $\gamma$-selectable OCRS.
    Initialize $I = \emptyset$, and let $C_i(k_i)$ be the event that $\sum_{i' \in I}k_{i'} \le K - k_i$ when element $e_i$ arrives, and given that its weight is $k_i$. When element $e_i \in N$ with weight $k_i$ arrives, if $\sum_{i' \in I}k_{i'} \le K - k_i$, we include $e_i$ in $I$ with probability $\frac{\gamma}{\Pr[C_i(k_i)]}$, where $\gamma  = \frac{1-k^*}{2-k^*}$. 
    The probability with which an element $e_i$ is selected is then:
 \[\sum_{k_i \in supp(\Kcal_i)} \Pr[\Kcal_i = k_i] \Pr\left[C_i(k_i) \right] \frac{\gamma}{ \Pr[C_i(k_i)]} = \gamma \sum_{k_i \in supp(\Kcal_i)} \Pr[\Kcal_i = k_i] = \gamma.\]
    It remains to prove that $\frac{\gamma}{ \Pr[C_i(k_i)]}$ is a valid coin (i.e., $\Pr[C_i(k_i)] \geq \gamma$). Let $W_{i}$ be the random variable that represents the total weight of elements in $I$ (i.e. $\sum_{i' \in I} k_{i'}$) when element $i$ arrives.
\begin{align*}
    \Ex{W_{i}} &= \sum_{i' < i} \sum_{k_{i'}  \in supp(\Kcal_{i'})} \Pr\left[C_{i'}(k_i) \right] \frac{\gamma}{ \Pr[C_{i'}(k_i)]} \Pr[\Kcal_{i'} = k_{i'}] k_{i'} \le \gamma K.
\end{align*}
Therefore, we have
\begin{align*}
    \Pr[C_{i}(k_i)] &= \Pr[E_{i} \le K - k_i] \\
    &= 1-\Pr[W_{i} > K - k_i] \\
    &\ge 1 - \frac{\gamma K}{K-k_i} \tag{Markov's Inequality}\\
    &\ge 1-  \frac{\gamma}{1- k^*} \\
    &= \gamma.
\end{align*}
This concludes the proof for the $\gamma$-selectable OCRS.

Next, we give a $\frac{1}{6}$-selectable OCRS.
With probability $1/2$ we run a ``heavy scheme,'' that only considers elements $e_i$ such that $k_i > \frac{K}{2}$; otherwise, we run a ``light scheme,'' that only considers elements $e_i$ such that $k_i \leq \frac{K}{2}$.

    
    Suppose we run the heavy scheme. Initialize $I = \emptyset$, and let $A_i$ be the event that $I = \emptyset$ when element $e_i$ arrives. For each element $e_i$ such that $k_i > \frac{K}{2}$, if $I = \emptyset$, we select $e_i$ with probability $\frac{1}{3 \Pr[A_i]}$. Assuming that $\frac{1}{3 \Pr[A_i]}$ is a valid coin (i.e., $\Pr[A_i] \geq 1/3$), the probability with which each element is selected, given that it is heavy and that we run the heavy scheme, is $\Pr[A_i] \frac{1}{3 \Pr[A_i]} = 1/3$. To prove that $\frac{1}{3 \Pr[A_i]}$ is a valid coin  we have:
    \begin{align*}
        \Pr[A_{i+1}] &= \Pr[A_{i}]\left(1 - \frac{1}{3 \Pr[A_i]} \sum_{\substack{k_i > K/2}} \Pr[\Kcal_i = k_i] \right) \\
        & = \Pr[A_{i}] - \frac{1}{3}\sum_{k_i > K/2} \Pr[\Kcal_i = k_i] \\
        &= 1- \frac{1}{3}\sum_{i' \leq i} \sum_{k_i > K/2} \Pr[\Kcal_i = k_i] \\
        &\ge 1-2/3 = 1/3. 
    \end{align*}

Now, suppose we run the light scheme. Notice that in this regime, where we ignore elements whose weight is larger than $K/2$, the previous $\gamma$-selectable OCRS is $1/3$-selectable (since $k^* = \frac{1}{K} \max_{i \in [n] , k_i \in supp(\Kcal_i)} \, k_i$). Since each scheme (heavy and light) is chosen with probability $1/2$, this OCRS is $1/6$-selectable.

This concludes the proof for the $\frac{1}{6}$-selectable OCRS.
\end{proof}

\begin{proposition}\label{prop: stochastic knapsack implementations}
    We can implement a $\left(c\left(\frac{1-\delta}{1+2\, \epsilon/c}\right)\right)$-selectable OCRS for the Stochastic Knapsack setting in time $poly(1/\epsilon^2, \log (1/\delta), n)$, where $c = \max\{\frac{1-k^*}{2-k^*}, 1/6\}$, and $k^* = \frac{1}{K} \max_{i \in [n] , k_i \in supp(\Kcal_i)} \, k_i$.
\end{proposition}

The proof of~\Cref{prop: stochastic knapsack implementations} simply combines the procedure outlined in the proof of \Cref{theorem:stochastic knapsack} and the analysis outlined in the proof of \Cref{proposition:efficient implimentation knapasack}. Details are deferred to~\Cref{sec: missing proofs from procurement}.

\section*{Acknowledgements}

Marios Mertzanidis and Alexandros Psomas are supported in part by an NSF CAREER award CCF-2144208, a Google AI for Social Good award, and research awards from Google and Supra.
M.\ Mertzanidis is also supported in part by NSF CCF 2209509 and 1814041. 
Divyarthi Mohan is supported in part by the European Research Council (ERC) under the European Union's Horizon 2020 research and innovation program (grant No. 866132). This work was partially conducted while D.\ Mohan was visiting the Simons Laufer Mathematical Sciences Institute (formerly MSRI) in Berkeley, California, during the Fall 2023 semester, which was funded by the National Science Foundation (grant No. DMS-1928930) and by the Alfred P. Sloan Foundation (grant G-2021-16778).

\bibliographystyle{alpha}
\bibliography{biblio}

\newcommand{\etalchar}[1]{$^{#1}$}
\begin{thebibliography}{MMPT24}

\bibitem[ABM{\etalchar{+}}22]{aggarwal2022simple}
Gagan Aggarwal, Kshipra Bhawalkar, Aranyak Mehta, Divyarthi Mohan, and
  Alexandros Psomas.
\newblock Simple mechanisms for welfare maximization in rich advertising
  auctions.
\newblock {\em Advances in Neural Information Processing Systems},
  35:28280--28292, 2022.

\bibitem[AGN18]{AnariGN18}
Nima Anari, Gagan Goel, and Afshin Nikzad.
\newblock Budget feasible procurement auctions.
\newblock {\em Operations Research}, 66(3):637--652, 2018.

\bibitem[Ala14]{alaei2014bayesian}
Saeed Alaei.
\newblock Bayesian combinatorial auctions: Expanding single buyer mechanisms to
  many buyers.
\newblock {\em SIAM Journal on Computing}, 43(2):930--972, 2014.

\bibitem[BCD20]{brustle2020multi}
Johannes Brustle, Yang Cai, and Constantinos Daskalakis.
\newblock Multi-item mechanisms without item-independence: Learnability via
  robustness.
\newblock In {\em Proceedings of the 21st ACM Conference on Economics and
  Computation}, pages 715--761, 2020.

\bibitem[BCKW15]{briest2015pricing}
Patrick Briest, Shuchi Chawla, Robert Kleinberg, and S~Matthew Weinberg.
\newblock Pricing lotteries.
\newblock {\em Journal of Economic Theory}, 156:144--174, 2015.

\bibitem[BGN17]{BabaioffGN17}
Moshe Babaioff, Yannai~A. Gonczarowski, and Noam Nisan.
\newblock The menu-size complexity of revenue approximation.
\newblock In {\em Proceedings of the 49th Annual {ACM} {SIGACT} Symposium on
  Theory of Computing, {STOC} 2017, Montreal, QC, Canada, June 19-23, 2017},
  pages 869--877, 2017.

\bibitem[BILW20]{babaioff2020simple}
Moshe Babaioff, Nicole Immorlica, Brendan Lucier, and S~Matthew Weinberg.
\newblock A simple and approximately optimal mechanism for an additive buyer.
\newblock {\em Journal of the ACM (JACM)}, 67(4):1--40, 2020.

\bibitem[Bor07]{border2007reduced}
Kim~C Border.
\newblock Reduced form auctions revisited.
\newblock {\em Economic Theory}, 31(1):167--181, 2007.

\bibitem[BS11]{bergemann2011robust}
Dirk Bergemann and Karl Schlag.
\newblock Robust monopoly pricing.
\newblock {\em Journal of Economic Theory}, 146(6):2527--2543, 2011.

\bibitem[CCD{\etalchar{+}}24]{chawla2024multi}
Shuchi Chawla, Dimitris Christou, Trung Dang, Zhiyi Huang, Gregory Kehne, and
  Rojin Rezvan.
\newblock A multi-dimensional online contention resolution scheme for revenue
  maximization.
\newblock {\em arXiv preprint arXiv:2404.14679}, 2024.

\bibitem[CD17]{cai2017learning}
Yang Cai and Constantinos Daskalakis.
\newblock Learning multi-item auctions with (or without) samples.
\newblock In {\em 2017 IEEE 58th Annual Symposium on Foundations of Computer
  Science (FOCS)}, pages 516--527. IEEE, 2017.

\bibitem[CDW12a]{CDW_1}
Yang Cai, Constantinos Daskalakis, and S~Matthew Weinberg.
\newblock An algorithmic characterization of multi-dimensional mechanisms.
\newblock In {\em Proceedings of the forty-fourth annual ACM symposium on
  Theory of computing}, pages 459--478, 2012.

\bibitem[CDW12b]{CDW_3}
Yang Cai, Constantinos Daskalakis, and S~Matthew Weinberg.
\newblock Optimal multi-dimensional mechanism design: Reducing revenue to
  welfare maximization.
\newblock In {\em 2012 IEEE 53rd Annual Symposium on Foundations of Computer
  Science}, pages 130--139. IEEE, 2012.

\bibitem[CDW13a]{CDW_2}
Yang Cai, Constantinos Daskalakis, and S~Matthew Weinberg.
\newblock Reducing revenue to welfare maximization: Approximation algorithms
  and other generalizations.
\newblock In {\em Proceedings of the Twenty-Fourth Annual ACM-SIAM Symposium on
  Discrete Algorithms}, pages 578--595. SIAM, 2013.

\bibitem[CDW13b]{CDW_4}
Yang Cai, Constantinos Daskalakis, and S~Matthew Weinberg.
\newblock Understanding incentives: Mechanism design becomes algorithm design.
\newblock In {\em 2013 IEEE 54th Annual Symposium on Foundations of Computer
  Science}, pages 618--627. IEEE, 2013.

\bibitem[CDW19]{cai2019duality}
Yang Cai, Nikhil~R Devanur, and S~Matthew Weinberg.
\newblock A duality-based unified approach to bayesian mechanism design.
\newblock {\em SIAM Journal on Computing}, 50(3):STOC16--160, 2019.

\bibitem[CGL11]{ChenGL11}
Ning Chen, Nick Gravin, and Pinyan Lu.
\newblock On the approximability of budget feasible mechanisms.
\newblock In {\em Proceedings of the Twenty-Second Annual ACM-SIAM Symposium on
  Discrete Algorithms}, SODA '11, page 685–699, USA, 2011. Society for
  Industrial and Applied Mathematics.

\bibitem[CHK07]{ChawlaHK07}
Shuchi Chawla, Jason~D. Hartline, and Robert Kleinberg.
\newblock Algorithmic pricing via virtual valuations.
\newblock In {\em Proceedings of the 8th ACM Conference on Electronic
  Commerce}, EC '07, pages 243--251, New York, NY, USA, 2007. ACM.

\bibitem[CHMS10]{ChawlaHMS10}
Shuchi Chawla, Jason~D Hartline, David~L Malec, and Balasubramanian Sivan.
\newblock Multi-parameter mechanism design and sequential posted pricing.
\newblock In {\em Proceedings of the forty-second ACM symposium on Theory of
  computing}, pages 311--320. ACM, 2010.

\bibitem[CM16]{chawla2016subadditive}
Shuchi Chawla and J.~Benjamin Miller.
\newblock Mechanism design for subadditive agents via an ex ante relaxation.
\newblock In {\em Proceedings of the 2016 ACM Conference on Economics and
  Computation}, EC '16, page 579–596, New York, NY, USA, 2016. Association
  for Computing Machinery.

\bibitem[CMS15]{ChawlaMS15}
Shuchi Chawla, David Malec, and Balasubramanian Sivan.
\newblock The power of randomness in bayesian optimal mechanism design.
\newblock {\em Games and Economic Behavior}, 91:297--317, 2015.

\bibitem[COVZ21]{cai2021efficient}
Yang Cai, Argyris Oikonomou, Grigoris Velegkas, and Mingfei Zhao.
\newblock An efficient $\epsilon$-bic to bic transformation and its application
  to black-box reduction in revenue maximization.
\newblock In {\em Proceedings of the 2021 ACM-SIAM Symposium on Discrete
  Algorithms (SODA)}, pages 1337--1356. SIAM, 2021.

\bibitem[CR14]{cole2014sample}
Richard Cole and Tim Roughgarden.
\newblock The sample complexity of revenue maximization.
\newblock In {\em Proceedings of the forty-sixth annual ACM symposium on Theory
  of computing}, pages 243--252, 2014.

\bibitem[CRTT23]{chawla2023buy}
Shuchi Chawla, Rojin Rezvan, Yifeng Teng, and Christos Tzamos.
\newblock Buy-many mechanisms for many unit-demand buyers.
\newblock In {\em International Conference on Web and Internet Economics},
  pages 21--38. Springer, 2023.

\bibitem[CVZ14]{chekuri2014submodular}
Chandra Chekuri, Jan Vondr\'ak, and Rico Zenklusen.
\newblock Submodular function maximization via the multilinear relaxation and
  contention resolution schemes.
\newblock {\em SIAM Journal on Computing}, 43(6):1831--1879, 2014.

\bibitem[CZ17]{cai2017subadditive}
Yang Cai and Mingfei Zhao.
\newblock Simple mechanisms for subadditive buyers via duality.
\newblock In {\em Proceedings of the 49th Annual ACM SIGACT Symposium on Theory
  of Computing}, STOC 2017, page 170–183, New York, NY, USA, 2017.
  Association for Computing Machinery.

\bibitem[Das15]{daskalakis2015multi}
Constantinos Daskalakis.
\newblock Multi-item auctions defying intuition?
\newblock {\em ACM SIGecom Exchanges}, 14(1):41--75, 2015.

\bibitem[DDT13]{daskalakis2013mechanism}
Constantinos Daskalakis, Alan Deckelbaum, and Christos Tzamos.
\newblock Mechanism design via optimal transport.
\newblock In {\em Proceedings of the fourteenth ACM conference on Electronic
  commerce}, pages 269--286, 2013.

\bibitem[DDT15]{daskalakis2015strong}
Constantinos Daskalakis, Alan Deckelbaum, and Christos Tzamos.
\newblock Strong duality for a multiple-good monopolist.
\newblock In {\em Proceedings of the Sixteenth ACM Conference on Economics and
  Computation}, pages 449--450, 2015.

\bibitem[DHKN21]{dughmi2021bernoulli}
Shaddin Dughmi, Jason Hartline, Robert~D Kleinberg, and Rad Niazadeh.
\newblock Bernoulli factories and black-box reductions in mechanism design.
\newblock {\em Journal of the ACM (JACM)}, 68(2):1--30, 2021.

\bibitem[DHP16]{devanur2016sample}
Nikhil~R Devanur, Zhiyi Huang, and Christos-Alexandros Psomas.
\newblock The sample complexity of auctions with side information.
\newblock In {\em Proceedings of the forty-eighth annual ACM symposium on
  Theory of Computing}, pages 426--439, 2016.

\bibitem[DK19]{Dutting19}
Paul D\"{u}tting and Thomas Kesselheim.
\newblock Posted pricing and prophet inequalities with inaccurate priors.
\newblock In {\em Proceedings of the 2019 ACM Conference on Economics and
  Computation}, EC '19, page 111–129, New York, NY, USA, 2019. Association
  for Computing Machinery.

\bibitem[DKL20]{duttingKL2020}
Paul D{\"u}tting, Thomas Kesselheim, and Brendan Lucier.
\newblock An o (log log m) prophet inequality for subadditive combinatorial
  auctions.
\newblock {\em ACM SIGecom Exchanges}, 18(2):32--37, 2020.

\bibitem[EFGT20]{ezra2020online}
Tomer Ezra, Michal Feldman, Nick Gravin, and Zhihao~Gavin Tang.
\newblock Online stochastic max-weight matching: Prophet inequality for vertex
  and edge arrival models.
\newblock In {\em Proceedings of the 21st ACM Conference on Economics and
  Computation}, pages 769--787, 2020.

\bibitem[FSZ21]{feldman2021online}
Moran Feldman, Ola Svensson, and Rico Zenklusen.
\newblock Online contention resolution schemes with applications to bayesian
  selection problems.
\newblock {\em SIAM Journal on Computing}, 50(2):255--300, 2021.

\bibitem[FTW{\etalchar{+}}21]{fu2021random}
Hu~Fu, Zhihao~Gavin Tang, Hongxun Wu, Jinzhao Wu, and Qianfan Zhang.
\newblock Random order vertex arrival contention resolution schemes for
  matching, with applications.
\newblock In {\em 48th International Colloquium on Automata, Languages, and
  Programming (ICALP 2021)}. Schloss-Dagstuhl-Leibniz Zentrum f{\"u}r
  Informatik, 2021.

\bibitem[GHZ19]{guo2019settling}
Chenghao Guo, Zhiyi Huang, and Xinzhi Zhang.
\newblock Settling the sample complexity of single-parameter revenue
  maximization.
\newblock In {\em Proceedings of the 51st Annual ACM SIGACT Symposium on Theory
  of Computing}, pages 662--673, 2019.

\bibitem[GJLZ20]{GravinJLZ20}
Nick Gravin, Yaonan Jin, Pinyan Lu, and Chenhao Zhang.
\newblock Optimal budget-feasible mechanisms for additive valuations.
\newblock {\em ACM Trans. Econ. Comput.}, 8(4), oct 2020.

\bibitem[GW21]{gonczarowski2021sample}
Yannai~A Gonczarowski and S~Matthew Weinberg.
\newblock The sample complexity of up-to-$\varepsilon$ multi-dimensional
  revenue maximization.
\newblock {\em Journal of the ACM (JACM)}, 68(3):1--28, 2021.

\bibitem[HKS07]{hajiaghayi2007automated}
Mohammad~Taghi Hajiaghayi, Robert Kleinberg, and Tuomas Sandholm.
\newblock Automated online mechanism design and prophet inequalities.
\newblock In {\em AAAI}, volume~7, pages 58--65, 2007.

\bibitem[HMR15]{huang2015making}
Zhiyi Huang, Yishay Mansour, and Tim Roughgarden.
\newblock Making the most of your samples.
\newblock In {\em Proceedings of the Sixteenth ACM Conference on Economics and
  Computation}, pages 45--60, 2015.

\bibitem[HN19]{hart2019selling}
Sergiu Hart and Noam Nisan.
\newblock Selling multiple correlated goods: Revenue maximization and menu-size
  complexity.
\newblock {\em Journal of Economic Theory}, 183:991--1029, 2019.

\bibitem[Hoe94]{hoeffding1994probability}
Wassily Hoeffding.
\newblock Probability inequalities for sums of bounded random variables.
\newblock {\em The collected works of Wassily Hoeffding}, pages 409--426, 1994.

\bibitem[HR15]{hart2015maximal}
Sergiu Hart and Philip~J Reny.
\newblock Maximal revenue with multiple goods: Nonmonotonicity and other
  observations.
\newblock {\em Theoretical Economics}, 10(3):893--922, 2015.

\bibitem[JMZ22]{jiang2022tight}
Jiashuo Jiang, Will Ma, and Jiawei Zhang.
\newblock Tight guarantees for multi-unit prophet inequalities and online
  stochastic knapsack.
\newblock In {\em Proceedings of the 2022 Annual ACM-SIAM Symposium on Discrete
  Algorithms (SODA)}, pages 1221--1246. SIAM, 2022.

\bibitem[KMS{\etalchar{+}}19]{KothariMSSW19}
Pravesh Kothari, Divyarthi Mohan, Ariel Schvartzman, Sahil Singla, and
  S.~Matthew Weinberg.
\newblock Approximation schemes for a unit-demand buyer with independent items
  via symmetries.
\newblock {\em 2019 IEEE 60th Annual Symposium on Foundations of Computer
  Science (FOCS)}, Nov 2019.

\bibitem[KO94]{keane1994bernoulli}
MS~Keane and George~L O'Brien.
\newblock A bernoulli factory.
\newblock {\em ACM Transactions on Modeling and Computer Simulation (TOMACS)},
  4(2):213--219, 1994.

\bibitem[KS22]{klumper2022budget}
Sophie Klumper and Guido Sch{\"a}fer.
\newblock Budget feasible mechanisms for procurement auctions with divisible
  agents.
\newblock In {\em International Symposium on Algorithmic Game Theory}, pages
  78--93. Springer, 2022.

\bibitem[KS23]{kashaev2023simple}
Danish Kashaev and Richard Santiago.
\newblock A simple optimal contention resolution scheme for uniform matroids.
\newblock {\em Theoretical Computer Science}, 940:81--96, 2023.

\bibitem[KW19]{kleinberg2019matroid}
Robert Kleinberg and S~Matthew Weinberg.
\newblock Matroid prophet inequalities and applications to multi-dimensional
  mechanism design.
\newblock {\em Games and Economic Behavior}, 113:97--115, 2019.

\bibitem[LLY19]{li2019revenue}
Yingkai Li, Pinyan Lu, and Haoran Ye.
\newblock Revenue maximization with imprecise distribution.
\newblock In {\em Proceedings of the 18th International Conference on
  Autonomous Agents and MultiAgent Systems}, pages 1582--1590, 2019.

\bibitem[LS18]{lee2018optimal}
Euiwoong Lee and Sahil Singla.
\newblock Optimal online contention resolution schemes via ex-ante prophet
  inequalities.
\newblock In {\em 26th Annual European Symposium on Algorithms (ESA 2018)}.
  Schloss Dagstuhl-Leibniz-Zentrum fuer Informatik, 2018.

\bibitem[MMPT24]{makur2024robustness}
Anuran Makur, Marios Mertzanidis, Alexandros Psomas, and Athina Terzoglou.
\newblock On the robustness of mechanism design under total variation distance.
\newblock {\em Advances in Neural Information Processing Systems}, 36, 2024.

\bibitem[MMZ24]{ma2024online}
Will Ma, Calum MacRury, and Jingwei Zhang.
\newblock Online contention resolution schemes for network revenue management
  and combinatorial auctions.
\newblock {\em arXiv preprint arXiv:2403.05378}, 2024.

\bibitem[Mor21]{morina2021bernoulli}
Giulio Morina.
\newblock Extending the bernoulli factory to a dice enterprise, January 2021.

\bibitem[MR16]{morgenstern2016learning}
Jamie Morgenstern and Tim Roughgarden.
\newblock Learning simple auctions.
\newblock In {\em Conference on Learning Theory}, pages 1298--1318. PMLR, 2016.

\bibitem[MV07]{manelli2007multidimensional}
Alejandro~M Manelli and Daniel~R Vincent.
\newblock Multidimensional mechanism design: Revenue maximization and the
  multiple-good monopoly.
\newblock {\em Journal of Economic theory}, 137(1):153--185, 2007.

\bibitem[Mye81]{myerson1981optimal}
Roger~B Myerson.
\newblock Optimal auction design.
\newblock {\em Mathematics of operations research}, 6(1):58--73, 1981.

\bibitem[NP05]{nacu2005fast}
Serban Nacu and Yuval Peres.
\newblock Fast simulation of new coins from old.
\newblock {\em The Annals of Applied Probability}, 15(1A):93--115, 2005.

\bibitem[NSW23]{srinivasan2023online}
Joseph~(Seffi) Naor, Aravind Srinivasan, and David Wajc.
\newblock Online dependent rounding schemes.
\newblock {\em arXiv preprint arXiv:2301.08680}, 2023.

\bibitem[PRSW22]{pollner2022improved}
Tristan Pollner, Mohammad Roghani, Amin Saberi, and David Wajc.
\newblock Improved online contention resolution for matchings and applications
  to the gig economy.
\newblock In {\em Proceedings of the 23rd ACM Conference on Economics and
  Computation}, pages 321--322, 2022.

\bibitem[PSCW22]{psomas2022infinite}
Alexandros Psomas, Ariel Schvartzman~Cohenca, and S~Weinberg.
\newblock On infinite separations between simple and optimal mechanisms.
\newblock {\em Advances in Neural Information Processing Systems},
  35:4818--4829, 2022.

\bibitem[PSW19]{psomas2019smoothed}
Alexandros Psomas, Ariel Schvartzman, and S~Matthew Weinberg.
\newblock Smoothed analysis of multi-item auctions with correlated values.
\newblock In {\em Proceedings of the 2019 ACM Conference on Economics and
  Computation}, pages 417--418, 2019.

\bibitem[RW15]{RubinsteinW15}
Aviad Rubinstein and S~Matthew Weinberg.
\newblock Simple mechanisms for a subadditive buyer and applications to revenue
  monotonicity.
\newblock In {\em Proceedings of the Sixteenth ACM Conference on Economics and
  Computation}, pages 377--394. ACM, 2015.

\bibitem[Sin10]{singer2010}
Yaron Singer.
\newblock Budget feasible mechanisms.
\newblock In {\em Proceedings of the 2010 IEEE 51st Annual Symposium on
  Foundations of Computer Science}, FOCS '10, page 765–774, USA, 2010. IEEE
  Computer Society.

\bibitem[Yao15]{Yao15}
Andrew Chi-Chih Yao.
\newblock An n-to-1 bidder reduction for multi-item auctions and its
  applications.
\newblock In {\em Proceedings of the Twenty-Sixth Annual ACM-SIAM Symposium on
  Discrete Algorithms}, pages 92--109. Society for Industrial and Applied
  Mathematics, 2015.

\end{thebibliography}

\appendix

\section{Further Related Work}\label{sec: related work}

\paragraph{Approximations in Bayesian Mechanism Design.} There is a rich body of literature on approximately optimal mechanism design, as discussed in the introduction. The work most related to ours is \cite{alaei2014bayesian} which provides a framework for designing multi-agent mechanisms by a reduction to single agent mechanisms via an ex-ante relaxation and using an online rounding procedure or OCRS. This framework yielded approximately optimal non-sequential and sequential mechanisms for several fundamental settings including selling $k$ identical items to $n$ additive buyers, who may be subject to additional constraints such as budget constraints, although does not support any inter-buyer constraints (except for per item supply constraints). In particular, a sequential posted price mechanism can obtain a $1-O(1/\sqrt{k})$ to the optimal revenue. In the special case of unit-demand buyers with independent values for heterogeneous items (where each item can be allocated to at most one bidder), \cite{alaei2014bayesian} establishes a sequential posted price mechanism that obtains a $4$-approximation---recovering a prior result of \cite{ChawlaHMS10}. Sequential mechanisms have been of particular interest due to their simplicity in implementation. The \cite{alaei2014bayesian} framework (along with other influential techniques like \cite{cai2019duality} duality) has enabled the design of approximately optimal and simple sequential mechanisms in various follow-up works. Chawla et al.~\cite{chawla2016subadditive} consider a more general case where agents have constrained additive valuations (aka matroid rank valuations) and provide a simple sequential two-part tariff  mechanism that obtains a constant approximation. This constant was then improved to $70$ by Cai at al.~\cite{cai2017subadditive}, who also extend their results to the setting where agents have XOS (and subadditive) valuations over $m$ heterogeneous, independent items, and \cite{duttingKL2020} provide improved approximation ratio of $O(\log\log m)$ for subadditive valuations. Crucially, these works assume that the agents' valuations are independent over the item. For the setting of correlated items and constrained additive buyers (with matroid constraints), Alai~\cite{alaei2014bayesian} provides a sequential mechanism that is a $2$-approximation to the optimal revenue. However, this mechanism is not ``simple" (unlike posted price or two-tariff). Indeed when the item distributions are correlated no ``simple'' mechanism can obtain any non-trivial approximation~\cite{hart2019selling,briest2015pricing}. A recent line of work circumvents this impossibilities by studying approximation guarantees of simple mechanisms with respect to a weaker benchmark called buy-many mechanisms~\cite{briest2015pricing,chawla2023buy,chawla2024multi}.

\paragraph{CRSs, OCRSs, and their applications.}

Contention resolution schemes were defined by Chekuri et al.~\cite{chekuri2014submodular}, and extended by Feldman et al.~\cite{feldman2021online} to online settings. CRSs and OCRSs have since found numerous applications in Bayesian and stochastic online optimization, such as stochastic probing~\cite{fu2021random}, prophet inequalities~\cite{ezra2020online} (in fact, OCRSs are equivalent to ex-ante prophet inequalities~\cite{lee2018optimal}), pricing problems~\cite{pollner2022improved,chawla2024multi}, and network revenue management~\cite{ma2024online}. In this paper, we define a type of dependent CRSs and OCRSs, which we call two-level CRSs/OCRSs. Online dependent rounding schemes, under the name of ODRS, were introduced by~\cite{srinivasan2023online} in the context of rounding fractional bipartite $b$-matchings online; here we give a different name to our schemes to highlight the specific dependence we can handle.

\section{Additional Preliminaries}

\subsection{Bernoulli Factories}\label{app: bernoulli}

We outline a few Bernoulli factories and their construction to help the reader gain some intuition:

\begin{enumerate}
    \item Bernoulli Negation: Given a $p$-coin, implement $f(p) = 1-p$. This can be implemented with one sample from the $p$-coin:
    \begin{itemize}
        \item $P \sim Bern[p]$.
        \item If $P=0$ output $1$ (otherwise output $0$).
    \end{itemize}
    \item Bernoulli Down Scaling: Given a $p$-coin, implement $f(p) = \lambda \cdot p$ for a constant $\lambda \in [0,1]$. This can be implemented with one sample from the $p$-coin:
    \begin{itemize}
        \item Draw $\Lambda \sim Bern[\lambda]$ and $P \sim Bern[p]$.
        \item Output $\Lambda \cdot P$.
    \end{itemize}
    \item Bernoulli Averaging: Given one $p_0$-coin and one $p_1$-coin, implement $f(p_0, p_1) = \frac{p_0+p_1}{2}$. This can be implemented with one sample from the $p_0$-coin or one sample from the $p_1$-coin:
    \begin{itemize}
        \item Draw $Z \sim Bern[1/2]$, $P_0 \sim Bern[p_0]$, and $P_1 \sim Bern[p_1]$.
        \item Output $P_Z$.
    \end{itemize}
    \item Bernoulli Doubling: Given a $p$-coin, implement $f(p) = 2p$ for $p \in (0, 1/2 - \delta]$. This can be implemented with $O(1/\delta)$ samples in expectation from the $p$-coin. The algorithm was introduced by Nacu et. al.~\cite{nacu2005fast}.
    \item Bernoulli Addition: Given one $p_0$-coin and one $p_1$-coin, implement $f(p_0, p_1) =  p_0+p_1$ for $p_0+p_1 \in [0, 1-\delta]$. This can be implemented with $O(1/\delta)$ samples in expectation from the $p_0$-coin and $p_1$-coin:
    \begin{itemize}
        \item Use Bernoulli Averaging to create a $\frac{p_0+p_1}{2}$-coin.
        \item Use Bernoulli Doubling on the $\frac{p_0+p_1}{2}$-coin.
    \end{itemize}
    \item Bernoulli Subtraction \cite{morina2021bernoulli}: Given one $p_0$-coin and one $p_1$-coin, implement $f(p_0, p_1) = p_1-p_0$ for $p_1-p_0 \ge \delta$. This can be implemented with $O(1/\delta)$ samples in expectation from the $p_0$-coin and $p_1$-coin:
    \begin{itemize}
        \item Use Bernoulli Negation on the $p_1$-coin to create a $1-p_1$-coin.
        \item Use Bernoulli Addition on the $1-p_1$-coin and $p_0$-coin to create a $1-p_1+p_0$-coin.
        \item Use Bernoulli Negation on the $1-p_1+p_0$-coin.
    \end{itemize}
\end{enumerate}

\begin{proof}[Proof of~\Cref{lem: bernoulli for division}]
Consider the following distribution:

\begin{equation*}
    \Pr[Y_k = y_k] = \begin{cases}
        1- \frac{1}{2}p_1 &\text{if } y_k = -1\\
        \frac{1}{2}(p_1-p_0) &\text{if } y_k = 0\\
        \frac{1}{2}p_0 &\text{if } y_k = 1
    \end{cases}
\end{equation*}

Then it is not difficult to see that:
\[\sum_{k=0}^{\infty} \Pr[Y_k = 1] \prod_{i=0}^{k-1} \Pr[Y_i = -1] = \frac{1}{2}p_0 \sum_{k=0}^{\infty}\left(1 - \frac{1}{2}p_1 \right)^k = \frac{p_0}{p_1},\]

and thus~\Cref{algorithm:bernoulli} is a valid Bernoulli factory for $p_0/p_1$.

    Let $N_t$ be the random variable that represents the number of tosses at round $t$, and let $X_t$ be a random variable that is $1$ if the experiment lasts at least $t$ rounds and $0$ otherwise. Then $N = \sum_{t =1}^{\infty} N_t \, X_t$. Linearity of expectation implies that $\Ex{N} =  \sum_{t =1}^{\infty} \Ex{N_t \, X_t}$, however  $X_t$ and $N_t$ are independent and thus $\Ex{N} =  \sum_{t =1}^{\infty} \Ex{N_t} \, \Ex{X_t}$. From, Morina~\cite{morina2021bernoulli} (Proposition 2.27) we know that $\Ex{N_t} \le 11.06(1+\delta^{-1})$. On the other hand $\sum_{t =1}^{\infty} \Ex{X_t} = \sum_{t =1}^{\infty} \left(1- \frac{1}{2}p_i \right)^t = \frac{2}{p_1}$. Combining the above concludes the proof.
\end{proof}

\section{Missing Proofs from Section~\ref{sec:constructing tCRS/tOCRS}}\label{app: missing from tCRS const}

\begin{proof}[Proof of~\Cref{theorem:knapsack unit demand}]
Let $(\Dcal,bx)$ be the two-level stochastic process through which elements become active. Let $k_{i,j}$ be the weight of element $e_{i,j}$, and $K$ bet the total weight. Our overall approach is similar to~\Cref{theorem:knapsack}.

Let $S_h = \{e_{i,j}: i\in [n], j \in [m]| k_{i,j}>K/2\}$ be the set of ``heavy elements'' and $S_{\ell} = \{e_{i,j}: i\in [n], j \in [m]| k_{i,j}\le K/2\}$ be the set of ``light elements.'' 
Our tOCRS is randomized: with probability $\frac{1+4b}{2+7b}$ we run a scheme that considers only heavy elements, the ``heavy scheme,'' and with probability $\frac{1+3b}{2+7b}$ we run a scheme that considers only light elements, the ``light scheme.''
In both cases, we use $I$ to indicate the set of elements we output. Without loss of generality (from the definition of a tOCRS) we assume that elements arrive in the order $e_{1,1}, e_{1,2}, \dots, e_{n,m}$.

For the heavy scheme, initialize $I= \emptyset$ and consider elements sequentially. Let $A_{i,j}(d_i)$ be the event that $I= \emptyset$ when element $e_{i,j}$ is considered when we run the heavy scheme and $d_i$ was sampled from the two-level stochastic process $(\Dcal,bx)$. If element $e_{i,j}$ is active, heavy, and $I= \emptyset$, then with probability $\frac{1}{(1+4b)\Pr[A_{i,j}(d_i)]}$ we set $I \leftarrow \{e_{i,j}\}$; otherwise we move on to the next element. The analysis in~\Cref{theorem:knapsack} proves that $\frac{1}{(1+4b)\Pr[A_{i,j}(d_i)]}$ is a valid probability. Notice that each heavy element is selected with probability exactly $\Pr[A_{i,j}(d_i)]\frac{1}{(1+4b)\Pr[A_{i,j}(d_i)]} = \frac{1}{(1+4b)}$ given that we run the heavy scheme.

Now consider the light case. We initialize $I= \emptyset$. Let $B_i$ be the event that $\sum_{\substack{i',j: e_{i',j} \in I, i'<i }} k_{i',j} < K/2$. Also let $C_{i,j}(d_i)$ be the event that, when $d_i$ was sampled from the two-level stochastic process $(\Dcal, bx)$, $\forall j' \in S_{\ell}, j' < j: e_{i,j'} \notin I$. We consider each light element $e_{i,j} \in S_{\ell}$ sequentially and, if it is active, and the weight of elements in $I$ is less than $K/2$, and no other element of $i$ has been selected, we set $I \leftarrow I \cup \{e_{i,j}\}$ with probability $\frac{1}{(1+3b)\Pr[B_{i} \cap C_{i,j}(d_i)]}$; otherwise we move on to the next element. If $\Pr[B_i \cap C_{i,j}(d_i)] \ge \frac{1}{(1+3b)}$ (i.e., if $\frac{1}{(1+3b)\Pr[B_{i} \cap C_{i,j}(d_i)]}$ is a valid probability), each light element is selected with probability exactly $\frac{1}{1+3b}$ when we run the light scheme. Towards proving that $\Pr[B_i \cap C_{i,j}(d_i)] \ge \frac{1}{(1+3b)}$, let $W_{i}$  elements in $I$ when we consider the first element of $i$. We have:

\begin{align*}
    \Ex{W_{i}} &= \sum_{i'<i} \sum_{d_{i'} \in \Vcal_{i'}} \sum_{j \in S_\ell}  \Pr[\Dcal_{i'} = d_{i'}] \Pr[B_{i'} \cap C_{i',j}(d_{i'})] \frac{b}{(3b+1)\Pr[B_{i'} \cap C_{i',j}(d_{i'})]} b x_{i',j}(d_{i'}) k_{i',j} \\
    &= \frac{b}{1+3b} \sum_{i'<i} \sum_{j \in S_\ell} w_{i',j} k_{i',j}\\
    &\le \frac{b}{1+3b}K. \tag{Knapsack feasibility of $(\Dcal,bx)$}
\end{align*}
Therefore:
\begin{align*}
    \Pr[B_{i}] &= \Pr[W_{i} < K/2] = 1-\Pr[W_{i} \ge K/2] \ge^{\text{(Markov's Inequality)}} 1 - \frac{\frac{b}{1+3b}K}{K/2} = \frac{1+b}{1+3b}.
\end{align*}
On the other hand,
\begin{align*}
    \Pr[C_{i,j+1}(d_{i})|B_i] &= \Pr[C_{i,j}(d_{i})|B_i]\left(1-\frac{1}{(1+3b)\Pr[B_{i} \cap C_{i,j}(d_{i})]}bx_{i,j}(d_i)\right)\\
    &= \Pr[C_{i,j}(d_{i})|B_i] - \frac{1}{(1+3b)\Pr[B_i]}bx_{i,j}(d_i)\\
    &= 1 - \frac{1}{(1+3b)\Pr[B_i]} \sum_{j'\leq j: j' \in S_\ell}bx_{i,j}(d_i) \\
    &\ge 1-\frac{b}{(1+3b)\Pr[B_i]}.\tag{multi-choice feasibility of $(\Dcal,bx)$}
\end{align*}
Combining the above we have that:
\begin{align*}
    \Pr[B_{i} \cap C_{i,j}(d_{i})] = \Pr[B_i] \Pr[C_{i,j+1}(d_{i})|B_i] \ge  \Pr[B_i] - \frac{b}{(1+3b)}= \frac{1}{1+3b}.
\end{align*}
We run the heavy scheme with probability $\frac{1+4b}{2+7b}$; thus, for an active element $e_{i,j} \in S_{h}$ we have:
\begin{align*}
    \Pr[e_{i,j} \in I] = \Pr[\text{``heavy scheme''}] \Pr[e_{i,j} \in I|\text{``heavy scheme''}] \ge \frac{1+4b}{2+7b}\frac{1}{1+4b} = \frac{1}{2+7b}.
\end{align*}
We can similarly show that active light elements are also selected with probability at least $ \frac{1}{2+7b}$. This concludes the proof of~\Cref{theorem:knapsack unit demand}. 
\end{proof}

\begin{proof}[Proof of~\Cref{prop: efficient knapsack unit demand}]
    Notice again that, similarly to~\Cref{proposition:efficient implimentation knapasack}, the only step we cannot directly implement from the procedure outlined in the proof of~\Cref{theorem:knapsack unit demand} is the toss of a $\Pr[C_{i,j+1}(d_i) \cap B_i]$-coin. Specifically, in the proof of~\Cref{theorem:knapsack unit demand} we do not calculate the following probability exactly:
    \[\Pr[C_{i,j+1}(d_i) \cap B_i] = \Pr[B_i] \Pr[C_{i,j+1}(d_i) |B_i] = \Pr[B_i]-\frac{b}{1+3b} \sum_{j'<j: j' \in S_l} x_{i,j}(d_i).\]
    
    Our procedure will again sequentially approximate these probabilities using multiple experiments, and bounding the error using Chernoff-Hoeffding bounds. 
    
    In order to decide whether to select some element $e_{i,j}$, we repeatedly simulate our algorithm until element $e_{i,j}$, for $T = \frac{1}{2\epsilon^2}\log \frac{2nm}{\delta}$ repetitions, where the choice of running the ``light scheme'' and $d_i$ are fixed. Let $X_t$ be the random variable that indicates if $B_{i}$ occurred at simulation $t \in [T]$; from Chernoff–Hoeffding bounds~\cite{hoeffding1994probability} we have that $\Pr\left[ \left|\frac{1}{T}\sum_{t \in [T]} X_t - \Pr[B_{i,j}(d_i)] \right|>\epsilon \right] \le \frac{\delta}{nm}$. Instead of selecting element $e_{i,j}$ (when it is active) with probability $\frac{1}{(1+3b)\Pr[B_{i} \cap C_{i,j}(d_i)]}$, we select it with probability $\frac{1}{(1+3b)\left(\ell+ \epsilon\right)}$-coin where $\ell = \frac{1}{T}\sum_{t \in [T]} X_t - \frac{b}{1+3b} \sum_{j'<j: j' \in S_l} x_{i,j}(d_i)$.

    Assuming that $\left|\frac{1}{T}\sum_{t \in [T]} X_t - \Pr[B_{i}] \right| \le \epsilon$ then $\left(\ell+ \epsilon \right) \in [\Pr[C_{i,j+1}(d_i) \cap B_i], \Pr[C_{i,j+1}(d_i) \cap B_i]+2\epsilon]$. Thus $\frac{\Pr[C_{i,j+1}(d_i) \cap B_i]}{(1+3b)\left(\ell+ \epsilon\right)} \le \frac{1}{1+3b}$. Thus $\Ex{W_i} \le \frac{b}{1+3b}K$ which gives us $\Pr[B_{i}]\ge \frac{1+b}{1+3b}$ and thus $\Pr[C_{i,j+1}(d_i) \cap B_i] \ge \frac{1}{1+3b} \ge 1/4$ using the same arguments presented in the proof of \Cref{theorem:knapsack unit demand}. 
     
    Thus, 
    \[\frac{\Pr[C_{i,j+1}(d_i) \cap B_i]}{(1+3b)\left(\ell+ \epsilon\right)} \ge \frac{1}{1+4b} \left(\frac{\Pr[C_{i,j+1}(d_i) \cap B_i]}{\Pr[C_{i,j+1}(d_i) \cap B_i]+2\epsilon} \right)  \ge \frac{1}{1+4b} \left(\frac{1}{1+8 \epsilon} \right).\]
    Union bounding we have that
    \[\Pr\left[ \exists (i,j) \in [n]\times [m]: \left|\frac{1}{T}\sum_{t \in [T]} X_t - \Pr[B_{i}] \right|>\epsilon \right] \le \delta.\]
    Thus with probability at least $1-\delta$ when we run the light scheme, each active element will be selected with probability at least $\frac{1}{1+4b} \left(\frac{1}{1+8 \epsilon} \right)$.
\end{proof}

\section{Missing from Section~\ref{sec: procurement}}

\subsection{Missing Preliminaries for Procurement Auctions}\label{app: missing procurement prelims}

A procurement auction elicits reported costs $(c_1, \dots, c_n)$, and determines which services are procured from which seller, as well as the payments to the sellers. A seller's objective is to maximize her expected utility, which is the total payment to her, minus the total cost she has to pay. A procurement auction is \emph{Bayesian Incentive Compatible (BIC)} if every seller $i \in [n]$ maximizes her expected utility by reporting her true costs $c_i$, assuming other sellers do so as well, where this expectation is over the randomness of other sellers' valuations, as well as the randomness of the mechanism. A mechanism is \emph{Bayesian Individually Rational} (BIR) if every seller $i \in [n]$ has non-negative expected utility when reporting her true cost (assuming other sellers do so as well). The (expected) value of a BIC procurement auction is the expected value the buyer makes when sellers draw their costs from $\Ccal$ (and report their true costs to the auction). We say that a procurement auction is BIC-IR if it is both BIC and BIR. A procurement auction is \emph{sequential} if it sequentially approaches each seller $i$, elicits a report, determines payments to seller $i$, and which services to procure from $i$, before proceeding to the next bidder. The \emph{optimal procurement auction} for a given distribution $\Ccal$, maximizes expected value over all BIC-IR procurement auctions. A procurement auction guarantees an $\alpha \geq 1$ approximation to the optimal value if its expected value is at least the expected value of the optimal procurement auction times $\frac{1}{\alpha}$.

The \emph{interim allocation} of a procurement auction $\M$, $\pi^{\M}$, indicates, for each seller $i$ and service $j$ the probability $\pi^{\M}_{i,j}(r_i)$ that seller $i$ receives service $j$ when she reports cost $r_i$ (over the randomness in $\M$ and the randomness in other sellers' reported costs $c_{-i}$, drawn from $\Ccal_{-i}$). The \emph{interim payment} of the buyer to seller $i$, $q^{\M}_{i}(r_i)$, is the expected payment she gets when she reports cost $r_i$ (again, over the randomness in $\M$ and the randomness in other sellers' reported costs).
The expected utility of seller $i$ with cost $c_i$ when reporting $r_i$ to a procurement auction $\M$, is $- \sum_{j \in [m]} c_{i,j} \pi^{\M}_{i,j}(r_i) + q^{\M}_{i}(r_i)$. An interim allocation rule $\pi$ is feasible in expectation if (i) $\forall i\in [n], c_i \in supp(\Ccal_i)$, $\pi_i(c_i)\in [0,1]$, and (ii) $\forall i \in [n], j \in [m]$, $\sum_{c_i \in supp(\Ccal_i)} \Pr[c_i] \cdot \pi_{i,j}(c_i) \leq 1$.

\subsection{Missing proofs}\label{sec: missing proofs from procurement}

\begin{proof}[Proof of~\Cref{prop: stochastic knapsack implementations}]
    The only step we cannot directly implement from the procedure outlined in the proof of~\Cref{theorem:stochastic knapsack} is the toss of the $\frac{\gamma}{\Pr[C_{i}(k_i)]}$ coin. The use of a Bernoulli factory would exponentially blow up the complexity of the procedure. Instead, we approximate these probabilities, sequentially, using multiple experiments and bounding the error using Chernoff bounds. 

    In order to decide whether to select some element $e_{i} \in N$, we repeatedly simulate our algorithm until element $e_{i}$, for $T = \frac{1}{2\epsilon^2}\log \frac{2|N|}{\delta}$ repetitions. In this simulation, the coins needed to make decisions until element $e_{i}$ are replaced with estimated coins (described shortly). Let $X_t$ be the indicator random variable for the event that $C_{i}(k_i)$ occurred at simulation $t \in [T]$. Instead of selecting element $e_{i}$ (when it is active) with probability $\frac{\gamma}{\Pr[C_{i}(k_i)]}$, we select it with probability $\frac{\gamma}{\left(\frac{1}{T}\sum_{t \in [T]} X_t+ \epsilon\right)}$.
    Standard Chernoff–Hoeffding bounds~\cite{hoeffding1994probability} imply that
    \begin{align*}
        \Pr\left[ \left|\frac{1}{T}\sum_{t \in [T]} X_t - \Pr[C_{i}(k_i)] \right|>\epsilon \right]& \le 2\exp\left(-2\epsilon^2 T \right) =2\exp\left(-2\epsilon^2 \frac{1}{2\epsilon^2}\log \frac{2|N|}{\delta} \right) \le \frac{\delta}{|N|}.
    \end{align*}
    Assuming that $\left|\frac{1}{T}\sum_{t \in [T]} X_t - \Pr[C_{i}(k_i)] \right| \le \epsilon$, $\left(\frac{1}{T}\sum_{t \in [T]} X_t+ \epsilon \right) \in [\Pr[C_{i}(k_i)], \Pr[C_{i}(k_i)]+2\epsilon]$. 
    
    When bounding $\Ex{W_i}$ (in the proof of~\Cref{prop: stochastic knapsack implementations}), the term $\frac{\gamma}{\Pr[C_{i}(k_i)]} \cdot \Pr[C_{i}(k_i)]$ is replaced with $\frac{\gamma \, \Pr[C_{i}(k_i)]}{\left(\frac{1}{T}\sum_{t \in [T]} X_t+ \epsilon\right)}$, which is at most $\gamma$. Thus $\Ex{W_i} \le \gamma \, K$, which gives us $\Pr[C_{i}(k_i)]\ge \gamma$ using the same arguments presented in the proof of \Cref{prop: stochastic knapsack implementations}. Thus, the probability that an active, element $e_{i}$ is selected is
    \[\frac{\Pr[\gamma C_{i}(k_i)]}{\left(\frac{1}{T}\sum_{t \in [T]} X_t+ \epsilon\right)} \ge \gamma \, \left(\frac{\Pr[C_{i}(k_i)]}{\Pr[C_{i}(k_i)]+2\epsilon} \right) \ge \gamma \, \left(\frac{1}{1+ 2 \epsilon/ \gamma} \right).\]
    Using a union bound we have that
    \[\Pr\left[ \exists i \in N: \left|\frac{1}{T}\sum_{t \in [T]} X_t - \Pr[C_{i}(k_i)] \right|>\epsilon \right] \le \delta.\]
    Thus, we overall have that with probability at least $1-\delta$, when we run the light scheme, each active element will be selected with probability at least $\gamma \left(\frac{1}{1+ 2 \epsilon/\gamma} \right)$.   
\end{proof}

\end{document}